%% file: MFCS2022.tex
\documentclass[a4paper,UKenglish,cleveref, autoref, thm-restate]{lipics-v2021}
\pdfoutput=1
\usepackage{amsfonts}
\usepackage{amsmath}
\usepackage{amssymb}
\usepackage{amsthm}
\usepackage{marvosym}
\usepackage{tikz}
\usepackage{textcomp}
\usepackage{stmaryrd}
\usepackage{wasysym}
\usepackage{soul}
\usepackage{multicol}

\usepackage[all]{xy}

\usepackage[T1]{fontenc}
\usepackage[utf8]{inputenc}

\usetikzlibrary{arrows, arrows.meta, automata, matrix, calc, positioning, shapes.geometric, decorations.pathreplacing}

\nolinenumbers

\title{Countdown $\mu$-calculus}

\author{Jędrzej Kołodziejski}{University of Warsaw, Poland}{j.kolodziejski@mimuw.edu.pl}{https://orcid.org/0000-0001-5008-9224}{National Science Center (NCN) grant 2021/41/B/ST6/00535}

\author{Bartek Klin}{University of Oxford, UK}{bartek.klin@cs.ox.ac.uk}{https://orcid.org/0000-0001-5793-7425}{}

\authorrunning{J. Kołodziejski and B. Klin} 

\Copyright{Jędrzej Kołodziejski and Bartek Klin} 



\relatedversion{Abridged version appeared in the proceedings of 47th International Symposium on Mathematical Foundations of Computer Science (MFCS 2022).}

\acknowledgements{We are grateful to Miko\l{}aj Boja\'nczyk for numerous helpful suggestions.}

\EventEditors{Stefan Szeider, Robert Ganian, and Alexandra Silva}
\EventNoEds{3}
\EventLongTitle{47th International Symposium on Mathematical Foundations of Computer Science (MFCS 2022)}
\EventShortTitle{MFCS 2022}
\EventAcronym{MFCS}
\EventYear{2022}
\EventDate{August 22--26, 2022}
\EventLocation{Vienna, Austria}
\EventLogo{}
\SeriesVolume{241}
\ArticleNo{55}

\newcommand{\A}{\mathcal{A}}

\newcommand{\D}{\mathcal{D}}
\newcommand{\G}{\mathcal{G}}
\newcommand{\HH}{\mathcal{H}}

\newcommand{\M}{\mathcal{M}}

\newcommand{\RR}{\mathcal{R}}
\newcommand{\SSS}{\mathcal{S}}
\newcommand{\U}{\mathcal{U}}

\newcommand{\letterA}{\mathsf{a}}
\newcommand{\letterB}{\mathsf{b}}
\newcommand{\arrowA}{\stackrel{\letterA}{\to}}
\newcommand{\arrowB}{\stackrel{\letterB}{\to}}
\newcommand{\arrowAction}{\stackrel{\action}{\to}}

\newcommand{\powersetNoArgs}{\mathcal{P}}
\newcommand{\powerset}[1]{\powersetNoArgs(#1)}

\newcommand{\Ord}{\mathsf{Ord}}
\newcommand{\ML}{\mathsf{ML}}

\newcommand{\MSO}{{\sf MSO}}
\newcommand{\MSOU}{{\sf MSO+U}}
\newcommand{\WMSOU}{{\sf WMSO+U}}

\newcommand{\muML}{\mu\text{-}\ML}
\newcommand{\muMLC}{\mu^\alpha\text{-}\ML}
\newcommand{\SubFor}{\mathsf{SubFor}}
\newcommand{\rank}{\mathsf{rank}}
\newcommand{\rangeRank}{\RR}
\newcommand{\nonstandardRanks}{\D}
\newcommand{\bounds}{\ctr_I}
\newcommand{\positions}{\mathsf{pos}}
\newcommand{\phase}{\SSS}
\newcommand{\phaseScalar}{\phase}
\newcommand{\phaseScalarBis}{\mathcal{Z}}
\newcommand{\val}{\mathsf{val}}
\newcommand{\ctr}{\mathsf{ctr}}
\newcommand{\eve}{\exists\text{ve}}
\newcommand{\adam}{\forall\text{dam}}

\newcommand{\game}{\G}

\newcommand{\semanticGame}{\game}
\newcommand{\exit}{\mathsf{exit}}
\newcommand{\exitPos}{Z}
\newcommand{\Var}{\mathsf{Var}}
\newcommand{\FreeVar}{\mathsf{FreeVar}}

\newcommand{\Actions}{\mathsf{Act}}
\newcommand{\action}{\tau}
\newcommand{\point}{\mathsf{m}}
\newcommand{\altpoint}{\mathsf{n}}

\newcommand{\exitEquiv}[1]{\equiv_{#1}}

\newcommand{\semantics}[1]{\llbracket #1 \rrbracket}

\renewcommand{\phi}{\varphi}
\renewcommand{\lim}{\text{lim}}
\renewcommand{\diamond}[1]{\langle #1 \rangle}
\newcommand{\boxmodal}[1]{[ #1 ]}
\renewcommand{\subset}{\subseteq}

\renewcommand{\max}{\mathsf{max}}

\newcommand{\worse}{\preccurlyeq}
\newcommand{\pconf}[1]{\langle#1\rangle}
\newcommand{\cconf}[1]{[#1]}

\newtheorem*{theorem-no-number}{Theorem}
\newtheorem*{lemma-no-number}{Lemma}

\theoremstyle{definition}

\begin{document}
\hideLIPIcs
\maketitle

\begin{abstract}
We introduce the countdown $\mu$-calculus, an extension of the modal $\mu$-calculus with ordinal approximations of fixpoint operators. In addition to properties definable in the classical calculus, it can express (un)boundedness properties such as the existence of arbitrarily long sequences of specific actions. The standard correspondence with parity games and automata extends to suitably defined countdown games and automata. However, unlike in the classical setting, the scalar fragment is provably weaker than the full vectorial calculus and corresponds to automata satisfying a simple syntactic condition. We establish some facts, in particular decidability of the model checking problem and strictness of the hierarchy induced by the maximal allowed nesting of our new operators.
\end{abstract}

\section{Introduction}

The modal $\mu$-calculus~\cite{Koz83} is a well-known logic for defining and verifying behavioural properties of state-and-transition systems. It extends propositional logic with basic next-step modalities and fixpoint operators to describe long-term behaviour. It is expressive enough to include other temporal logics such as CTL* as fragments, but it has good computational properties, and
its simple syntax and semantics makes it a convenient formalism to study.

%
The $\mu$-calculus has a straightforward inductively-defined semantics, but it is often useful to consider an alternative (but equivalent) semantics based on parity games. A formula $\phi$ together with a model $\M$ define a game between two players called $\adam$ and $\eve$. Positions in the game are of the form $(\point,\psi)$ where $\point$ is a point in $\M$ and $\psi$ is a subformula of $\phi$, and moves are defined so that $\eve$ has a winning strategy from $(\point,\phi)$ if and only if $\phi$ holds in $\point$. Among other advantages, the game-based semantics provides more efficient algorithms for model checking of $\mu$-calculus formulas than an inductive computation of fixpoints~\cite{CJKLS17}.

The model component can be abstracted away from parity games. Indeed, a formula $\phi$ itself gives rise to an alternating parity automaton $\A_\phi$ that recognizes models. The behaviour of an automaton on a model is defined in terms of a parity game, states of $\A_\phi$ are subformulas of $\phi$, and the transition relation is defined so that it accepts a model $\M$ rooted in a point $\point$ if and only if $\phi$ holds in $\point$. The advantage of this is that $\A_\phi$, while conceptually closer to a parity game, is a finite structure even if it is then applied to infinite models.

The modal $\mu$-calculus is a rather expressive formalism: it can define all bisimulation-invariant properties definable in monadic second-order logic (\MSO)~\cite{JW96}, such as ``there is an infinite path of $\tau$-labeled edges''. However, there are some properties of interest which are not definable even in \MSO. Notable examples include (un)boundedness properties such as ``for every number $n$, there is a path with at least  $n$ consecutive $\tau$-labeled edges''. An extension of \MSO{} called \MSOU{}, aimed at defining such properties, has been considered~\cite{Boj04}. However, the satisfiability problem of \MSOU{} turned out to be undecidable even for word models~\cite{BPT16}. Since the modal $\mu$-calculus is a fragment of \MSO{}, it is worthwhile to extend it with a mechanism for defining (un)boundedness properties, in the hope of retaining decidability.

In this paper we propose such an extension: the {\em countdown $\mu$-calculus} $\muMLC$. In addition to $\mu$-calculus operators, it features countdown operators $\mu^\alpha$ and $\nu^\alpha$ parametrized by ordinal numbers $\alpha$. Instead of least and greatest fixpoints, they define ordinal approximations of those fixpoints. Intuitively, while the meaning of classical $\mu$-calculus formulas $\mu x.\phi(x)$ and $\nu x.\phi(x)$ is defined by infinite unfolding of the formula $\phi$ until a fixpoint is reached, for $\mu^\alpha x.\phi(x)$ and $\nu^\alpha x.\phi(x)$ the unfolding stops after $\alpha$ steps (which makes a difference if $\alpha$ is smaller than the \emph{closure ordinal} of $\phi$). The classical fixpoint operators are kept but renamed to $\mu^\infty$ and $\nu^\infty$, to make clear the lack of any restrictions on the unfolding process.

An inductive definition of the semantics of countdown formulas is just as straightforward as in the classical case. With some more effort, we are able to formulate game-based semantics as well. We introduce {\em countdown games} and {\em countdown automata}, which are similar to parity games and alternating automata known from the classical setting, but are additionally equipped with counters that are decremented and reset by the two players according to specific rules. Intuitively, the counters say how many more times various ranks can be visited, in similar manner to the signatures introduced by Walukiewicz \cite[Section 3]{Wal01}. A player responsible for decrementing a counter may lose the game if the value of that counter is zero, just as a player responsible for finding the next position in a game may lose if there is no position to go to. The key mechanism of countdown games is implicit in \cite{HK17}, where the authors investigate a nonstandard semantics for the scalar fragment of the $\mu$-calculus equivalent to replacing every $\mu$ and $\nu$ by our countdown operators $\mu^\alpha$ and $\nu^\alpha$, respectively. However, the authors do not abstract from formulas in their definition of games, nor consider the full vectorial calculus that corresponds to automata.

A correspondence between countdown formulas, automata and games is as tight as for the classical $\mu$-calculus. However, complications arise: the distinction between {\em vectorial} and {\em scalar} formulas, which in the classical case disappears to a large extent due to the so-called Beki\'c principle, now becomes pronounced. We prove that vectorial countdown calculus is more expressive than its scalar fragment. We also prove that the countdown operator nesting hierarchy of formulas is proper.

We conjecture that the satisfiability problem is decidable for $\muMLC$. Unfortunately, the lack of positional determinacy in countdown games prevents us from using proof techniques known from parity automata (where one can transform an alternating automaton into a nondeterministic one that guesses the positional strategy). Nevertheless, the existence of an automata model equivalent to logic is encouraging. Apart from allowing us to solve some fragments of the logic, it implies that $\muMLC$ does not share some of the troublesome properties of $\MSOU$ that result in undecidability. In particular, it can be used to show that all languages definable in $\muML$ have \emph{bounded topological complexity} (i.e. at most $\Sigma^1_2$, see \cite{Skrzyp16} for an introduction to topological methods in computer science). Since $\MSOU$ defines a $\Sigma^1_n$-complete language for every $n<\omega$ \cite[Theorem 2.1]{HS12}, \cite[Theorem 7]{Skrzyp16}, it follows that some $\MSOU$-definable languages are not expressible in $\muMLC$ (whether $\muMLC$-definability implies $\MSOU$-definability remains an open question). Since by \cite[Theorem 1.3]{BDGPS20bis} every logic closed under boolean combinations, projections and defining the language $U$ from Example \ref{Ex-Language-U} contains $\MSOU$, this means that our calculus is \emph{not closed under projections}. This is an arguably good news, as in the light of \cite[Theorem 1.4]{BKSZ20}, giving up closure under projections is the only way to go if one wants to design a decidable extension of $\MSO$ closed under boolean operations. Decidability of the weak variant $\WMSOU$ of $\MSOU$ over infinite words \cite{Boj11} and infinite (ranked) trees \cite{BT12} shows that such extensions are possible. In fact, both results are obtained by establishing a correspondence with equivalent automata models, namely deterministic max-automata \cite[Theorem 1]{Boj11} and nested limsup automata \cite[Theorem 2]{BT12}. Since the existence of accepting runs for such automata can be expressed in $\muMLC$, we get that $\muMLC$ contains $\WMSOU$ on infinite words and trees. The opposite inclusion is false (due to topological reasons), at least for the trees. The relation between $\muMLC$ and the $\omega B$-, $\omega S$- and $\omega BS$-automata of \cite{BC06} remains unclear, as these models do not admit determinization. Also, the relation between our logic and regular cost functions (see e.g. \cite{Colc13}) is less immediate than it could seem at first glance and requires further research.

\section{Preliminaries}\label{section 2}

\noindent{\bf Fixpoints.} Let $\Ord$ be the class of all ordinals, and $\Ord_\infty$ the class $\Ord$ extended with an additional element $\infty$ greater than all ordinals.

Knaster-Tarski theorem says that every monotonic function $F:A \to A$ on a complete lattice $A$ has the least and the greatest fixpoint, which we denote $F^\infty_\mu$ and $F^\infty_ \nu$. Moreover:
    \begin{itemize}
      \item $F^\infty_\mu$ is the limit of the increasing sequence $F^\alpha_\mu = \bigvee_{\beta<\alpha}F(F^\beta_\mu)$
      \item $F^\infty_ \nu$ is the limit of the decreasing sequence $F^\alpha_\nu = \bigwedge_{\beta<\alpha}F(F^\beta_\nu)$
    \end{itemize}  
  where $\alpha \in \Ord$ and $\bigvee,\bigwedge$ are the join and meet operations in $A$.

\medskip
\noindent{\bf Parity games.} 
A \emph{parity game} is played between two players $\eve$ and $\adam$ (or simply $\exists$ and $\forall$). It consists of a set of {\em positions} $V = V_\exists \sqcup V_\forall$ divided between both players, an edge relation $E \subseteq V \times V$, and a labeling $\rank:V \to \rangeRank$ for some finite linear order $\rangeRank=\rangeRank_\exists\sqcup\rangeRank_\forall$ divided between the two players. 

A {\em play} is a sequence of positions. After a play $\pi = v_1\ldots v_n\in V^*$, the owner of $v_n$ chooses $(v_n,v_{n+1})\in~E$ and the game moves to $v_{n+1}$. A player who has no legal moves loses immediately. To determine the winner of an infinite play, we look at the highest $r \in \rangeRank$ such that positions with rank $r$ appear infinitely often in the play, and the owner of $r$ loses. 

A \emph{strategy} for a player $P\in\{\exists,\forall\}$ is a partial map $\sigma:V^*V_P \to E$ that tells the player how to move. A play $v_1v_2\ldots$ is {\em consistent} with $\sigma$ if for every $n$ such that $v_n\in V_P$ we have $\sigma(v_1\ldots v_n)=v_{n+1}$. A strategy $\sigma$ is \emph{winning} from a position $v$ if every play that begins in $v$ and is consistent with $\sigma$ is a win for $P$. A strategy is \emph{positional} if $\sigma(\pi)$ depends only on the last position in $\pi$. Parity games are \emph{positionally determined}: if a player has a winning strategy from $v$ then (s)he has a winning positional strategy.

\medskip
\noindent{\bf Modal $\mu$-calculus.}
A model $\M$ for a fixed set $\Actions$ of atomic {\em actions} consists of a set of {\em points} $M\ni\point,\altpoint,\cdots$ together with a binary relation ${\arrowAction}\subset M\times M$ for every $\action\in\Actions$. 

Formulas of the modal $\mu$-calculus $\muML$ are given by the grammar:
\begin{equation}\label{eq:classicalmu}
  \phi ::=  x\ |\ \top\ |\ \bot\ |\ \phi_1\vee\phi_2\ |\ \phi_1\wedge\phi_2\ |\ \mu x.\phi\ |\ \nu x.\phi\ |\ \diamond{\action}\phi\ |\ \boxmodal{\action}\phi
\end{equation}
where $x$ ranges over a fixed infinite set $\Var$ of variables and $\action\in\Actions$. Given a valuation $\val:\Var\to\powerset{M}$, the semantics $\semantics{\phi}^\val\subseteq M$ for all formulas $\phi$ is defined inductively, with $\mu x.\phi$ and $\nu x.\phi$ denoting the least and greatest fixpoints, respectively, of the monotonic function $H\mapsto\semantics{\phi}^{\val[x\mapsto H]}$ on the complete lattice $\powerset{M}$. More details can be found e.g.~in~\cite{AN01,Ven20}, but they can also be discerned from Section~\ref{sec:ctdmu} below, where the semantics of countdown $\mu$-calculus is presented in detail.

The above syntax does not include negation, but $\mu$-calculus formulas are semantically closed under negation. For every formula $\phi$ there is a formula $\widetilde{\phi}$ that acts as the negation of $\phi$ on every model, defined by induction in a straightforward way:
\begin{align}\label{eq:neg}
	\widetilde{\phi_1 \vee \phi_2} = \widetilde{\phi_1} \wedge \widetilde{\phi_2} ,\qquad
	 \widetilde{\diamond{\action}\phi} = \boxmodal{\action}\widetilde{\phi}, \qquad
	 \widetilde{\mu x.\phi} = \nu x.\widetilde{\phi}, \qquad \text{etc.}
\end{align}

\medskip
\noindent{\bf Vectorial $\mu$-calculus.} A syntactically richer version of the modal $\mu$-calculus admits mutual fixpoint definitions of multiple properties, in formulas such as
$
	\mu_1 (x_1,x_2).(\phi_1,\phi_2),
$ 
where variables $x_1$ and $x_2$ may occur both in $\phi_1$ and $\phi_2$. Given a valuation $\val$ as before, this formula is interpreted as the least fixpoint of the monotonic function $(H_1,H_2)\mapsto (\semantics{\phi_1}^{\val[x_i\mapsto H_i]},\semantics{\phi_1}^{\val[x_i\mapsto H_i]})$ on the complete lattice $\powerset{M}^2$; the resulting pair of sets is then projected to the first component as dictated by the subscript in $\mu_1$. Tuples of any size are allowed. This {\em vectorial} calculus is expressively equivalent to the scalar version described before, thanks to the so-called {\em Beki\'c principle} which says that the equality:
\begin{align}\label{eq:bekic}
  \mu
  \begin{pmatrix}
    x_1 \\
    x_2
  \end{pmatrix}
  .
  \begin{pmatrix}
    f_1(x_1,x_2) \\
    f_2(x_1,x_2)
  \end{pmatrix}
  =
  \begin{pmatrix}
    \mu x_1.f_1(x_1,\ \mu x_2.f_2(x_1,x_2)) \\
    \mu x_2.f_2(\mu x_1.f_1(x_1,x_2),\ x_2)
  \end{pmatrix}
\end{align}
holds for every pair of monotone operations $f_i:A_1\times A_2 \to A_i$ on complete lattices $A_1,A_2$, and similarly for the greatest fixpoint operator $\nu$ in place of $\mu$. 

\section{Countdown $\mu$-calculus}\label{sec:ctdmu}

We now introduce the \emph{countdown $\mu$-calculus} $\muMLC$. We begin with the scalar version.

\subsection{The scalar fragment}

As before, fix an infinite set $\Var$ of variables and a set $\Actions$ of actions. The syntax of \emph{(scalar) countdown $\mu$-calculus} is defined as follows:
\begin{equation}\label{eq:scalarmu}
  \phi ::=  x\ |\ \top\ |\ \bot\ |\ \phi_1\vee\phi_2\ |\ \phi_1\wedge\phi_2\ |\ \mu^\alpha x.\phi\ |\ \nu^\alpha x.\phi\ |\ \diamond{\action}\phi\ |\ \boxmodal{\action}\phi
\end{equation}
for $x\in\Var$, $\action \in \Actions$ and $\alpha \in \Ord_\infty$; the presence of ordinal numbers $\alpha$ is the only syntactic difference with~\eqref{eq:classicalmu}. A formula with no free variables is called a \emph{sentence}. In case $|\Actions|=1$, we may skip the labels and write $\Diamond$ and $\Box$ instead of $\diamond{\action}$ and $\boxmodal{\action}$. In statements that apply both to least and greatest fixpoints, we will sometimes use $\eta^\alpha$ to denote either $\mu^\alpha$ or $\nu^\alpha$.

Given a model $\M$, for every valuation $\val: \Var \to \powerset{M}$, the \emph{semantics} $\semantics{\phi}^\val \subset M$ is defined inductively as follows:
\begin{align*}
    \semantics{x}^\val = \val(x);&\\
    \semantics{\top}^\val = M
    \text{\ \ \ and\ \ \ }&
    \semantics{\bot}^\val = \emptyset\\
    \semantics{\phi_1 \vee \phi_2}^\val = \semantics{\phi_1}^\val \cup \semantics{\phi_2}^\val
    \text{\ \ \ and\ \ \ }&
    \semantics{\phi_1 \wedge \phi_2}^\val = \semantics{\phi_1}^\val \cap \semantics{\phi_2}^\val;\\
    \semantics{\diamond{\action}\phi}^\val = \{\point \in M\ |\ \exists_{\altpoint \in \semantics{\phi}^\val}\ \point \arrowAction \altpoint\}
    \text{\ \ \ and\ \ \ }&
    \semantics{\boxmodal{\action}\phi}^\val = \{\point \in M\ |\ \forall_{\altpoint \in \semantics{\phi}^\val}\ \point \arrowAction \altpoint\};\\
    \semantics{\mu^\alpha x.\phi}^\val = F^\alpha_\mu
    \text{\ \ \ and\ \ \ }&
    \semantics{\nu^\alpha x.\phi}^\val = F^\alpha_\nu
\end{align*}
where in the last clause $F(H)=\semantics{\phi}^{\val[x \mapsto H]}$. 
We will skip the index $\val$ if it is immaterial or clear from the context.

This obviously contains the classical $\mu$-calculus, but is capable of capturing \emph{boundedness} and \emph{unboundedness} properties which are not expressible in the classical setting:

\begin{example}\label{ex:aaaa}
  For $|\Actions|=1$,  consider the formula $\nu^\alpha x . \Diamond x$. In a model $\M$, for $\alpha<\omega$ the set $\semantics{\nu^\alpha x . \Diamond x}$ consists of the points from which there is a path of length at least $\alpha$. Hence, $\nu^\omega x . \Diamond x$ holds in a point if there are arbitrarily long finite paths starting from there.
\end{example}


\subsection{The vectorial calculus}

The (full) \emph{countdown $\mu$-calculus} is defined as for its scalar fragment, except that fixpoint operators act on tuples (vectors) of formulas rather than on single formulas.
\begin{definition}
  The syntax of \emph{countdown $\mu$-calculus} is given as follows:
  \[
    \phi ::=  x\ |\ \top\ |\ \bot\ |\ \phi_1\vee\phi_2\ |\ \phi_1\wedge\phi_2\ |\  \mu^\alpha_i \overline{x}.\overline{\phi}\ |\ \nu^\alpha_i \overline{x}.\overline{\phi}\ |\ \diamond{\action}\phi\ |\ \boxmodal{\action}\phi
  \]
for $1\leq i \leq n<\omega$, $\overline{x} = \langle x_1, ..., x_n\rangle \in \Var^n$, $\overline{\phi} = \langle \phi_1, ..., \phi_n \rangle$ a tuple of formulas, $\action \in \Actions$ and $\alpha \in \Ord_\infty$.
\end{definition}


\begin{definition}\label{def:vec-semantics}
  The meaning $\semantics{\phi}^\val\subset M$ of a formula $\phi$ in a model $\M$ under valuation $\val$ is defined by induction the same way as for the scalar formulas except for the operators $\mu^\alpha_i$ and $\nu^\alpha_i$, in which case:
  \[
    \semantics{\mu^\alpha_i \overline{x}.\overline{\phi}}^\val = \pi_i(F^\alpha_\mu) \text{\ \ \ and\ \ \ } \semantics{\nu^\alpha_i \overline{x}.\overline{\phi}}^\val = \pi_i(F^\alpha_\nu)
  \]
  where the monotone map $F: {(\powerset{M})}^n \to {(\powerset{M})}^n$ is given as:
  \[
    F(H_1, ..., H_n) = (\semantics{\phi_1}^{\val'}, ..., \semantics{\phi_n}^{\val'})
  \]
  for $\val' = \val[x_1 \mapsto H_1, ..., x_n \mapsto H_n]$ and $\pi_i:{(\powerset{M})}^n \to \powerset{M}$ is the $i$-th projection.
\end{definition}
Note that operators $\mu^\infty$ and $\nu^\infty$ are equivalent to $\mu$ and $\nu$ from the classical $\mu$-calculus. 
Furthermore, for every ordinal $\alpha$, the formula $\mu^{\alpha+1}_i\overline{x}.\overline{\psi}$ is equivalent to \[
\psi_i[x_1 \mapsto \mu^\alpha_1\overline{x}.\overline{\psi},\ldots, x_n \mapsto \mu^\alpha_n\overline{x}.\overline{\psi}]
\]
and similarly for $\nu^{\alpha+1}$. As a result, without loss of generality we may assume that in countdown operators $\mu^\alpha$ and $\nu^\alpha$ only limit ordinals $\alpha$ are used.

The countdown $\mu$-calculus is semantically closed under negation in the same way as the classical calculus, extending~\eqref{eq:neg} with the straightforward
$\widetilde{\mu^\alpha_i \overline{x}.\overline{\phi}} = \nu^\alpha_i \overline{x}.\overline{\widetilde{\phi}}$ and  $\widetilde{\nu^\alpha_i \overline{x}.\overline{\phi}} = \mu^\alpha_i \overline{x}.\overline{\widetilde{\phi}}$.

In Section~\ref{Sec-vectorial-scalar} we will compare the expressive power of the vectorial and scalar countdown $\mu$-calculus in detail. For now, let us show that Beki\'c principle~\eqref{eq:bekic} fails for countdown operators:

\begin{example}\label{Ex-Language-U}
  An infinite word $W\in\Gamma^\omega$ over the alphabet $\Gamma=\{\letterA,\letterB\}$ can be seen as a model for $\Actions=\Gamma$ with $\omega$ as the set of points and with transition relations defined by:
  \[
  	n\arrowAction m \iff m=n+1 \text{ and }W_n=\action.
  \]
For every regular language $K\subset \Gamma^*$ and $x \in \Var$, it is straightforward to define a fixpoint formula (in the classical $\mu$-calculus, so without countdown operators) $\diamond{K}x$ that holds in a point $n$, for a valuation $\val$, if and only if there exists a word $w \in K$ and a path in $W$ labelled with $w$ that starts in $n$ and ends in a point that belongs to $\val(x)$. Then, the formula: 
  \[
    \phi = \nu^\omega_1(x_1, x_2) . (\diamond{\Gamma^*}x_2, \diamond{\letterA} x_2)
  \]
  is true in a word $W$ iff it contains arbitrarily long blocks of consecutive $\letterA$'s. To see this, observe that at the $i$-th step of approximation: (i) the second component ($x_2$) contains a point $n$ iff the next $i$ transitions are all labelled with $\letterA$, and (ii) the first component ($x_1$) contains a point $n$ iff the second component contains at least one point after $n$.
  
  However, the following scalar formula constructed by analogy to the Beki\'c principle:
  \[
    \psi = \nu^\omega x_1 . \diamond{\Gamma^*}(\nu^\omega x_2. \diamond{\letterA} x_2)
  \]
is equivalent to $\diamond{\Gamma^*}( \nu^\omega x_2. \diamond{\letterA} x_2)$, and the formula under $\diamond{\Gamma^*}$ holds in a point iff all the future transitions from that point are labelled with $\letterA$. Thus, $\psi$ holds (in any point) iff the word $W$ is of the form $\Gamma^* \letterA^\omega$, and so $\psi$ is not equivalent to $\phi$.
\end{example}

\section{Countdown Games}

The notion of a countdown game extends that of a parity game. As for parity games, it assumes a fixed finite linear order of ranks $\rangeRank=\rangeRank_\exists\sqcup\rangeRank_\forall$. In addition, we fix a subset $\nonstandardRanks\subseteq\rangeRank$ of {\em nonstandard} ranks; at positions with these ranks countdowns will occur. Denote $\nonstandardRanks_\exists=\nonstandardRanks\cap\rangeRank_\exists$ and $\nonstandardRanks_\forall=\nonstandardRanks\cap\rangeRank_\forall$. 

  A \emph{countdown game} consists of a set of {\em positions} $V = V_\exists \sqcup V_\forall$ divided between players $\eve$ and $\adam$, an edge relation $E\subseteq V \times V$, a labelling $\rank: V \to \rangeRank$, and an initial counter valuation $\bounds: \nonstandardRanks \to \Ord$. Each nonstandard rank has an associated counter.
  
Each game {\em configuration} consists of a position $v \in V$ together with a counter valuation $\ctr: \nonstandardRanks \to \Ord$. We consider {\em positional} and {\em countdown} configurations, denoted respectively $\pconf{v,\ctr}$ and $\cconf{v,\ctr}$, with the following moves allowed:
  \begin{itemize}
    \item From a positional configuration $\pconf{v, \ctr}$, the owner of $v$ chooses an edge $(v,w) \in E$ and the game proceeds from the countdown configuration $\cconf{w, \ctr}$;
    \item From a countdown configuration $\cconf{v, \ctr}$, the owner of $r=\rank(v)$ chooses a counter valuation $\ctr'$ such that:
    \begin{itemize}
    	\item $\ctr'(r') = \bounds(r')$ for $r'<r$,
	\item $\ctr'(r) < \ctr(r)$ (if $r$ is nonstandard),
	\item $\ctr'(r') = \ctr(r')$ for $r'>r$,
    \end{itemize}
    and the game proceeds from the positional configuration $\pconf{v,\ctr'}$. In words: counters for ranks lower than $r$ are reset, the counter for $r$ (if any) is decremented, and counters for higher ranks are left unchanged. Note that if $r$ is standard then there is no real choice here: $\ctr'$ is determined by $\ctr$. And if $r$ is nonstandard then the move amounts to choosing an ordinal $\alpha<\ctr(r)$.
  \end{itemize}
  
Every play of the game alternates between positional and countdown configurations, and in each move only one component of the configuration is modified. Therefore, although a play is formally a sequence of configurations, it can be more succinctly represented as an alternating sequence of positions and counter valuations:
\begin{align}\label{eq:play}
	\pi = v_1\ctr_2v_2\ctr_2v_3\ctr_3\cdots
\end{align}
This has the same length as the sequence of configurations, and we will call it the length of the play. A \emph{phase} of a game is a set of its finite plays that is convex with respect to the prefix ordering.
Given a phase $\phase$ and a play $\pi\in\phase$, we denote by $\phase_\pi$ the subset of $\phase$ consisting of all the plays having $\pi$ as a prefix.
 
In any configuration, if the player responsible for making the next move is stuck, (s)he looses immediately. Otherwise, in an infinite play, the owner of the greatest rank appearing infinitely often looses, as in parity games. Strategies and winning strategies are defined as for classical parity games, as partial functions from finite plays to moves.

Given configuration $\gamma$, we denote the game \emph{initialized in the configuration $\gamma$} by $\game,\gamma$. The default initial counter assignment is $\bounds$ and the default initial mode is the positional one, meaning that $\game,v$ stands for $\game,\pconf{v,\bounds}$. 
%

Note that the only way the counters may interfere with a play is when a counter has value $0$ and so its owner cannot decrement it. It is therefore beneficial for a player to have greater ordinals at his/her counters.

Countdown games are not positionally determined, in the sense that the players may need to look at the counter values in order to choose a winning move (although they are \emph{configurationally determined}, since a countdown game $\game$ can be seen as a parity game with configurations of $\game$ as its positions). Later, we will show how to upgrade strategies to enforce a very limited form of counter-independence.

\section{Countdown Automata}

Countdown automata are a stepping stone between formulas and games. A countdown formula will define an automaton, which will then recognize a model in terms of a countdown game. Since formulas can have free variables, for technical reasons we will also consider automata with free variables. These variables resemble terminal states in that they can be targets of transitions, but no transitions originate in them, and whether they accept or not depends on an external valuation.

\begin{definition}
  A \emph{countdown automaton} consists of:
  \begin{itemize}
    \item a finite set of states $Q = Q_\exists \sqcup Q_\forall$ divided between two players;
    \item an initial state $q_I \in Q$;
    \item a transition function $\delta: Q \to \powerset{Q\sqcup \Var} \sqcup (\Actions \times (Q\sqcup \Var))$
    (we call the left part \emph{$\epsilon$-transitions} and the right one \emph{modal transitions});
    \item an assignment of ranks $\rank : Q \to \rangeRank$ and an assignment of initial counter values $\bounds:\nonstandardRanks \to \Ord$, as in a countdown game.
\end{itemize}
\end{definition}

The language of an automaton is defined in terms of a countdown game, analogously to parity games and parity automata.

\begin{definition}\label{def:semantic-game}
  Fix an automaton $\A = (Q, q_I, \delta, \rank, \bounds)$. Given a model $\M$, a valuation $\val:\Var\to\powerset{M}$ and a point $\point_I \in M$, we define the semantic game $\semanticGame^\val(\A)$ to be the countdown game $(V,E,\rank', \bounds 
  )$ where positions are of the form $V = M \times (Q\sqcup\Var)$ 
  and the edge relation $E$ is defined as follows. In a position $(\point,q)$ for $q\in Q$:
    \begin{itemize}
      \item if $\delta(q)\subset Q\sqcup\Var$, outgoing edges (called $\epsilon$-edges, or $\epsilon$-moves) are $
        \{((\point,q),(\point,z))\ |\ z\in\delta(q)\}$,
      \item if $\delta(q) = (\action, p)$, outgoing edges (modal edges, modal moves) are
      $
        \{((\point,q),(\altpoint,p))\ |\ \point \arrowAction \altpoint\}.
      $
    \end{itemize}
    There are no outgoing edges from positions $(\point,x)$ for $x\in \Var$.
    
    For $q\in Q$, the owner of the position $(\point,q)$ is the owner of the state $q$, and $\rank'(\point,q)=\rank(q)$. For $x\in \Var$, the position $(\point,x)$ belongs to $\adam$ if $\point\in\val(x)$ and to $\eve$ otherwise. The rank $\rank'(\point,x)$ can be set arbitrarily, as it does not affect the outcome of the game. The initial counter assignment $\bounds$ is kept the same. 
    
    The language $\semantics{\A}^\val\subset M$ of an automaton $\A$ is the set of all points $\point\in M$ for which the configuration $\pconf{(\point,q_I),\bounds}$ in the game $\semanticGame^\val(\A)$ is winning for $\eve$.
\end{definition}

It is worth to mention that although in general countdown games are not positional, one can show a much weaker but still useful fact: in the particular case of semantic games, the winning player always has a strategy that does not look at the counters in the initial \emph{pre-modal} phase of the game (that is, \emph{before the first modal move}). The precise statement can be found in Proposition~\ref{Prop-PreModal-CtrIndep} in Appendix~\ref{app:guarded}; its game-theoretic core is Proposition~\ref{Prop-Finite-CtrIndep} in Appendix~\ref{app:games}.

The countdown calculus and countdown automata have the same expressive power, i.e. there are language-preserving translations $\phi \mapsto \A_\phi$ and $\A \mapsto \phi_\A$ between formulas and automata. As in the classical setting, the link between formulas and automata is very useful in establishing facts about the logic. For example, one can use game semantics to show that every formula of the standard $\muML$ can be transformed into an equivalent guarded one. Thanks to the equivalence between \emph{countdown} formulas and \emph{countdown} automata, the same is true for $\muMLC$, as stated in Proposition~\ref{Prop-Guardedness} in Appendix~\ref{app:guarded}.

We will now explain the translations between logic and automata in turn.

\subsection{From formulas to automata -- Game Semantics}\label{sec:formtoaut}

    Every countdown formula $\phi\in\muMLC$ gives rise to a countdown automaton $\A_\phi$ such that $\semantics{\phi}^\val=\semantics{\A_\phi}^\val$ for every model $\M$ and valuation $\val$.
   Specifically, given a formula $\phi$ (with some free variables), we define an automaton $\A_\phi=(Q,q_I,\delta,\rank,\bounds)$ (over the same free variables) as follows:
    \begin{itemize}
        \item $Q=\SubFor(\phi)-\FreeVar(\phi)$ is the set of all subformulas other than the free variables of $\phi$ (\emph{without} identifying different occurrences of identical subformulas, i.e., here a subformula means a path in the syntactic tree of $\phi$ from the root of $\phi$ to the root node of the subformula). Ownership of a state in $Q$ depends on the topmost connective, with $\eve$ owning $\vee$ and $\diamond{\action}$ and $\adam$ owning $\wedge$ and $\boxmodal{\action}$; ownership of fixpoint subformulas, countdown subformulas and variables can be set arbitrarily as it will not matter;
        \item $q_I=\phi$;
        \item the transition function is defined by cases:
        \begin{itemize}
        		\item $\delta(\theta_1\vee\theta_2)=\delta(\theta_1\wedge\theta_2) = \{\theta_1,\theta_2\}$,
		\item $\delta(\diamond{\action}\theta)=\delta(\boxmodal{\action}\theta)=(\action,\theta)$,
		\item $\delta(\eta^\alpha_i \overline{x}.\overline{\theta})=\{\theta_i\}$ (for $\eta=\mu$ or $\eta=\nu$),
		\item $\delta(x) = \{\theta_i\}$, where $\eta^\alpha_j (x_1, ..., x_n).(\theta_1, ..., \theta_n)$ is the (unique) subformula of $\phi$ binding $x$ with $x=x_i$.
        \end{itemize}
        \item For the ranking function, assume that the lowest rank in $\rangeRank$ is standard and call it $0$ (ownership of this rank does not matter). Then let $\rank$ assign $0$ to all subformulas of $\phi$ except for immediate subformulas of fixpoint operators. To those, assign ranks in such a way that subformulas have strictly smaller ranks than their superformulas, and for every subformula $\eta^\alpha_i \overline{x}.\overline{\phi}$:
        \begin{itemize}
            \item all formulas in the tuple $\overline{\phi}$ have the same rank $r$,
            \item $r$ belongs to $\eve$ if $\eta = \mu$ and to $\adam$ if $\eta=\nu$, and
            \item if $\alpha=\infty$ then $r$ is standard, otherwise it is nonstandard and $\bounds(r)=\alpha$.
        \end{itemize}
    \end{itemize}
    We denote $\semanticGame^\val(\phi)=\semanticGame^\val(\A_\phi)$.


\begin{theorem}[Adequacy]\label{Thm-Adequacy-Countdown}
  For every model $\M$ and valuation $\val$, $\semantics{\phi}^\val = \semantics{\A_\phi}^\val$.
\end{theorem}
\begin{proof} As with the classical $mu$-calculus, the proof proceeds by induction on the complexity of the formula. The only new cases of $\mu^\alpha\overline{x}.\overline{\phi}$ and $\nu^\alpha\overline{x}.\overline{\phi}$ are proven by transfinite induction on $\alpha$. For the details, see Appendix~\ref{app:formtoaut}.
\end{proof}

\begin{example}
For $\Actions=\{\action\}$, consider the formula $\phi = \nu^\omega x . \Diamond x$ from Example~\ref{ex:aaaa}. The automaton $\A_\phi$ has three states: $Q = \{\phi,{\Diamond x},x\}$, with $\phi$ the initial state, and the transition function comprises two deterministic $\epsilon$-transitions and one modal transition:
\[
	\delta(\phi) = \{{\Diamond x}\}, \qquad
	\delta({\Diamond x}) = (\tau,x), \qquad
	\delta(x) = \{{\Diamond x}\}.
\]
The state ${\Diamond x}$ is owned by $\eve$; ownership of the other two states does not matter. The automaton uses two ranks, $0<1$, where $0$ is standard and $1$ is nonstandard, assigned to states by:
$\rank(\phi)=\rank(x)=0$ and $\rank({\Diamond x}) = 1$.
Rank $1$ is owned by $\adam$; ownership of rank $0$ does not matter. (Note how the state ${\Diamond x}$ is owned by $\eve$, but its rank is owned by $\adam$). The initial counter value is $\bounds(1)=\omega$.

Now consider any model $\M$. Since $\Actions$ has only one element, $\M$ is simply a directed graph. The semantic game $\semanticGame(\phi)$ on $\M$ ($\phi$ has no free variables, so neither has $\A_\phi$ and we need not consider valuations $\val$) has positions of the form $(\point,q)$ where $\point \in M$ and $q\in Q$, with ownership and rank inherited from $q$. Edges are of the form:
\begin{itemize}
\item $((\point,\phi),(\point,{\Diamond x}))$ and $((\point,x),(\point,\phi))$ -- the $\epsilon$-edges,
\item $((\point,{\Diamond x}),(\altpoint,x))$ such that $\point\to\altpoint$ is an edge in $\M$ -- the modal edges.
\end{itemize}
Configurations of the game arise from positions together with counter valuations; there is only one nonstandard rank, so a counter valuation is simply an ordinal. 

For a point $\point\in\M$, the default initial configuration of the game is the positional configuration $\pconf{(\point,\phi),\omega}$. A play that begins in this configuration proceeds as follows:
\begin{enumerate}
\item The first move is deterministic, to the countdown configuration $\cconf{(\point,{\Diamond x}),\omega}$.
\item $\adam$, as the owner of the rank of ${\Diamond x}$, makes the next move: he chooses a number $k<\omega$, and the games moves to the positional configuration $\pconf{(\point,{\Diamond x}),k}$.
\item $\eve$ owns the position, so she makes the next move: she chooses a point $\altpoint\in M$ such that $\point\arrowAction\altpoint$, and the game moves to the countdown configuration 
$\cconf{(\altpoint,x),k}$.
\item The rank of $x$ is standard, so in the next move the counter does not change and the game moves to $\pconf{(\altpoint,x),k}$. The next move is also deterministic, to the countdown configuration
$\cconf{(\altpoint,{\Diamond x}),k}$. The game then goes back to step 2.~above, with $k$ in place of $\omega$.
\end{enumerate}
From this it is clear that $\eve$ wins from $\pconf{(\point,\phi),\omega}$ if and only if $\M$ has arbitrarily long paths that begin in $\point$, as stated in Example~\ref{ex:aaaa}.
\end{example}

\subsection{From automata to formulas}\label{sec:auttoform}

\begin{theorem}\label{thm:auttoform}
    For every countdown automaton $\A$ there exists a countdown formula $\phi_\A$ s.t. $\semantics{\A}^\val=\semantics{\phi_\A}^\val$ for every model $\M$ and valuation $\val$.
\end{theorem}

\begin{proof}
See Appendix~\ref{app:auttoform}; here we just sketch the construction of $\phi_\A$.
For an automaton $\A=(Q,q_I,\delta,\rank,\bounds)$, by induction on $r\in \rangeRank$ we build a formula $\psi_{r,q}$ for each $q\in Q$. Then we put $\phi_\A = \psi_{r_{\max},q_I}$. Thus for the base case of the lowest rank $r=0$:
  \begin{itemize}
    \item if $\delta(s)=(\action,p)$ then for $\psi_{0,s}$ we put $\diamond{\action}x_p$ if $q$ belongs to $\eve$ and $\boxmodal{\action}x_p$ if $q$ belongs to $\adam$, 
    \item if $\delta(s)\subset Q$ then for $\psi_{0,s} $ we put  $\bigvee_{p\in\delta(s)}x_p$ if $q$ belongs to $\eve$ and $\bigwedge_{p\in\delta(s)}x_p$ if $q$ belongs to $\adam$.
  \end{itemize}
 
  For the inductive step, let $q_1,...,q_d$ be all states in $Q$ with rank $r$. For every $q_i$ define the vectorial formula:
  \[
    \theta_i=\eta^\alpha_{q_i}(x_{q_1},...,x_{q_d}).(\psi_{r,q_1},...,\psi_{r,q_d})
  \]
  with $\alpha=\bounds(r)$ and $\eta=\mu$ if $r$ belongs to $\eve$ and $\eta=\nu$ if $r$ belongs to $\adam$. Then put $\psi_{r+1,q}= \psi_{r,q}[x_{q_1}\mapsto\theta_1, ..., x_{q_d}\mapsto\theta_d]$.
\end{proof}





\section{Vectorial vs. scalar calculus}\label{Sec-vectorial-scalar}

In this section we investigate the relation between scalar and vectorial formulas. We have already seen with Example \ref{Ex-Language-U} that unlike with standard fixpoints, the Beki\'{c} principle is not valid in the countdown setting. Interestingly, scalar formulas correspond to automata with a simple syntactic restriction.

\begin{proposition}\label{Prop-Scalar=Injectively-Ranked}
  Scalar countdown formulas and automata where every two states have different ranks have equal expressive power.
\end{proposition}

\begin{proof}
  Inspecting the translations between formulas and automata from Sections~\ref{sec:formtoaut} and~\ref{sec:auttoform}, it is evident that injectively ranked automata are translated to scalar formulas, and that, although in our translation the choice of the assignment of ranks is not deterministic, every scalar formula can be translated to an injectively ranked automaton.
\end{proof}

Since the Beki\'{c} principle fails, a natural question is whether there is another way of transforming vectorial formulas to scalar form (or, equivalently, arbitrary countdown automata to injectively ranked ones). We shall give a negative answer in Theorem \ref{Thm-Vectorial-Vs-Scalar}. However, before we proceed, let us analyse the following example, which shows that scalar formulas are more expressive than they may seem, covering in particular the property from Example~\ref{Ex-Language-U}.

\subsection{Languages of unbounded infixes}

  Fix a regular language of finite words $L \subset \Gamma^*$. Let $\U(L) \subset \Gamma^\omega$ be the language of all infinite words that contain arbitrarily long infixes from $L$. For instance, the language from Example \ref{Ex-Language-U} is $\U(\letterA^*)$. We shall now show that $\U(L)$ can be defined in the countdown $\mu$-calculus, first by a vectorial formula, then by a scalar one.
  
Consider a finite deterministic automaton $\A=(Q,\delta,q_I,F)$ that recognizes $L$. Let $\delta^+:\Gamma^+\times Q \to Q$ be the unique inductive extension of the transition function $\delta:\Gamma\times Q\to Q$ to nonempty words. Define $K_{p,q}=\{w \in \Gamma^+\ |\ \delta^+(w,p)=q\}$ the (regular) language of nonempty words leading from $p$ to $q$ in $\A$, and let $K_{p,F}$ denote the union $\bigcup_{q\in F}K_{p,q}$. By the pigeonhole principle we have
  $
    \U(L)=\bigcup_{q\in Q} \U_q(L)
  $,
  where $\U_q(L) \subset \Gamma^\omega$ consists of words such that for every $n<\omega$, $w$ has an infix $w_n=v_Iu_1...u_nv_F \in L$ s.t. (i) $v_I \in K_{q_I,q}$, (ii) $u_1,...,u_n\in K_{q,q}$, and (iii) $v_F\in K_{q,F}$. Then $\U_q(L)$ can be defined by a vectorial formula:
  \[
    \U_q(L)=\semantics{\nu^\omega_1(x_1,x_2).(\diamond{\Gamma^*K_{q_I,q}}x_2, \diamond{K_{q,q}} x_2 \wedge \diamond{K_{q,F}}\top)}
  \]
  where $\diamond{K}\psi$ is the formula as explained in Example~\ref{Ex-Language-U}. Indeed, the corresponding semantic game on a word $w$ proceeds as follows:
  \begin{enumerate}
    \item $\adam$ chooses a number $n<\omega$ as the value of his only counter, 
    \item $\eve$ skips a prefix $v_0v_I\in\Gamma^*K_{q_I,q}$ of $w$, 
    \item $\adam$ decrements his counter;
    \item $\eve$ keeps moving through $u_1, u_2, ... \in K_{q,q}$ so that after each step, some state in $F$ is reachable from $q$ by some prefix of the remaining word. After each such choice of $u_i$ $\adam$ has to decrement his counter, and so $\eve$ wins iff she can make at least $n-1$ such steps.
  \end{enumerate}
  The two different stages in which $\adam$'s counter is decremented reflect the two-phase dynamics of the game: first $\adam$ challenges $\eve$ with a number, and then $\eve$ shows that she can provide an infix long enough. 
  
  It is more tricky to define the language $\U_q(L)$ with a scalar formula, but it turns out to be possible. To this end, observe that without loss of generality we may restrict attention to words $w$ such that:
  \begin{enumerate}
    \item the infixes $w_n \in L$ start arbitrarily far in $w$;
    \item each $w_n$ can be decomposed as $v_Iu_1...u_nv_F \in L$ s.t. (i) $v_I \in K_{q_I,q}$, (ii) $u_1,...,u_n\in K_{q,q}$, (iii) $v_F\in K_{q,F}$, and additionally \emph{(iv) all $u_i$ begin with the same letter $\letterA \in \Gamma$};
    \item there are at least two distinct letters $\letterA,\letterB \in \Gamma$ that appear infinitely often in $w$;
    \item the first letter of $w$ is $\letterB$.
  \end{enumerate}
  Indeed, for (1) note that otherwise $w_n$ start in the same position $k$ for all $n$ large enough. But then even the stronger property ``There exists a position $k$ such that the run of $\A$ from $k$ visits $q$ and $F$ infinitely often'' holds, and this is easily definable by a fixpoint formula.

  Item (2) follows from the pigeonhole principle and the observation that in $w_{n\times |\Gamma|}=v_Iu_1...u_{n\times |\Gamma|}v_F$ at least $n$ $u_i$'s begin with the same letter.
  
  For (3) observe that otherwise $w$ has a suffix $\letterA^\omega$ for some $\letterA \in \Gamma$, in which case membership in $\U_q(L)$ is definable by a fixpoint formula. This is because an ultimately periodic word is bisimilar to a finite model, and so every monotone map reaches its fixpoints in finitely many steps, meaning that the countdown operator $\nu^\omega$ is equivalent to $\nu^\infty$.
  
  Finally, for (4) note that the language $\U_q(L)$ is closed under adding and removing finite prefixes, and so if a formula $\phi$ defines $\U_q(L)\cap \letterB\Gamma^\omega$, then the formula $\diamond{\Gamma^*}(\diamond{\letterB}\top\wedge\phi)$ defines $\U_q(L)$.

  With this in mind, define:
  \[
    \phi = \nu^\omega x.(\diamond{\letterB}\top \wedge \diamond{\Gamma^*K_{q_I,q}}(\diamond{\letterA}\top \wedge x))
    \vee
    (\diamond{K_{q,q}}(\diamond{\letterA}\top \wedge x) \wedge \diamond{\letterA}\top \wedge \diamond{K_{q,F}}\top).
  \]
  Note how $\diamond{\letterB}\top \wedge x$ and $\diamond{\letterA}\top\wedge x$ replace $x_1$ and $x_2$ from the vectorial formula. Consider the corresponding semantic game on a word $w$. Consider configurations of the game with the main disjunction as the formula component. Every infinite play of the game must visit such configurations infinitely often. In such a configuration, if the next letter in the model is either $\letterA$ or $\letterB$ then $\eve$ must choose the right or left disjunct respectively. In particular, once the game reaches a configuration where $\diamond{\letterA}\top$ holds, it must also hold every time the variable $x$ in unraveled in the future.
  As a result, $\eve$ wins from a configuration where $\diamond{\letterA}\top$ holds against $\adam$'s counter $n<\omega$ iff there is $u_1...u_{n+1}v_F$ starting in the current position such that $u_1,...,u_{n+1} \in K_{q,q}$, each $u_i$ starts with $\letterA$, and $v_F \in K_{q,F}$. Moreover, $\eve$ wins from a position where $\diamond{\letterB}\top$ holds, against $\adam$'s $n+1<\omega$, iff there is $v_I \in \Gamma^* K_{q_I,q}$ starting in the current position such that the next position after $v_I$ satisfies $\diamond{\letterA}\top$ and $\eve$ wins from there against $n$. Putting this together, we get that $\eve$ wins from a position satisfying $\diamond{\letterB}\top$ against $n$ iff there is $v_Iu_1...u_nv_F = w_n$ as in condition (2) above. Since the game starts with $\adam$ choosing an arbitrary $n<\omega$, it follows that indeed $\phi$ defines $\U_q(L)$.  

\subsection{Greater expressive power of the vectorial calculus}

We now show an example of a property that is definable in the vectorial countdown calculus but not in the scalar one. 

Fixing $\Actions = \{\letterA,\letterB\}$, consider a model $\M = (M,\arrowA,\arrowB)$ with points $M=\{\point_i,\altpoint_i\ |\ i <\omega\}$, and with exactly the edges: $\point_i \arrowA \point_j$,  $\altpoint_i \arrowA \point_j$  and $\altpoint_i \arrowB \point_j$ for all $i>j$; and $\point_i \arrowB \point_j$ for all $i$ and $j$.
Note that the relation $\arrowA$ is a subset of $\arrowB$. The model is shown in Fig.~\ref{Fig-Model}.

\begin{figure}
\hspace*{2cm}
\begin{tikzpicture}[every initial by arrow/.style={->, thick, scale = 0.6}]
\newcommand{\pictureScale}{0.65}
\newcommand{\horizontalStep}{5*\pictureScale}
\newcommand{\verticalGap}{5*\pictureScale}

\definecolor{colorNode}{RGB}{241, 178, 225}
\definecolor{colorArrow1}{RGB}{8,39,245}
\definecolor{colorArrow2}{RGB}{228,0,110}

\tikzstyle{main node} = [scale=(\pictureScale*1.4, circle, draw, fill=colorNode, minimum size = \pictureScale*1.6cm]
\tikzstyle{dotBox} = [draw=none, font={\LARGE$\cdot$}]
\tikzstyle{initial} = [initial text = {}, initial distance = \pictureScale*1.25cm]
\tikzstyle{arrowA+B} = [-stealth', color=colorArrow1, line width=\pictureScale*1.0mm]
\tikzstyle{arrowB} = [-stealth', color=colorArrow2, line width=\pictureScale*0.75mm]

\node[main node]                                  (v^+_0) at ($(0, 0)+ (0,\verticalGap)$)    {$\point_0$};
\node[main node]                                  (v^+_1) at ($(v^+_0)-(\horizontalStep,0)$) {$\point_1$};
\node[main node]                                  (v^+_2) at ($(v^+_1)-(\horizontalStep,0)$) {$\point_2$};

\node[dotBox]                                     (dot0+) at ($(v^+_2)-(\pictureScale*2.8,0)$) {};
\node[dotBox]                                     (dot1+) at ($(dot0+)-(\pictureScale*1.7,0)$) {};
\node[dotBox]                                     (dot2+) at ($(dot1+)-(\pictureScale*1.7,0)$) {};

\node[main node]                                  (v^-_0) at ($(0, 0)$)    {$\altpoint_0$};
\node[main node]                                  (v^-_1) at ($(v^-_0)-(\horizontalStep,0)$) {$\altpoint_1$};
\node[main node]                                  (v^-_2) at ($(v^-_1)-(\horizontalStep,0)$) {$\altpoint_2$};

\node[dotBox]                                     (dot0-) at ($(v^-_2)-(\pictureScale*2.8,0)$) {};
\node[dotBox]                                     (dot1-) at ($(dot0-)-(\pictureScale*1.7,0)$) {};
\node[dotBox]                                     (dot2-) at ($(dot1-)-(\pictureScale*1.7,0)$) {};

\path[every node/.style={font=\sffamily\small}]
     (v^+_2) edge [arrowA+B, out = 350, in = 195] (v^+_1)
     (v^+_1) edge [arrowA+B, out = 350, in = 190] (v^+_0)
     (v^+_2) edge [arrowA+B, out = 345,  in = 214] (v^+_0)

     (v^-_2) edge [arrowA+B, bend left = 10, out = 20, in = 180] (v^+_1)
     (v^-_2) edge [arrowA+B, bend left = 15, out = 10, in = 213] (v^+_0)
     (v^-_1) edge [arrowA+B, bend left = 15, out = 350, in = 220] (v^+_0)

     (v^+_0) edge [arrowB, loop, out = 0, in = 60,   min distance = \pictureScale*37] (v^+_0)
     (v^+_1) edge [arrowB, loop, out = 65, in = 125, min distance = \pictureScale*35] (v^+_1)
     (v^+_2) edge [arrowB, loop, out = 75, in = 135, min distance = \pictureScale*35] (v^+_2)

     (v^+_0) edge [arrowB, out = 170, in = 15] (v^+_1)
     (v^+_1) edge [arrowB, out = 165, in = 10] (v^+_2)
     (v^+_0) edge [arrowB, out = 150, in = 35] (v^+_2)
;
\end{tikzpicture}

\caption{The model $\M$. Blue arrows represent edges labeled both with $\letterA$ and $\letterB$, and pink arrows are edges labeled only with $\letterB$.}\label{Fig-Model}
\end{figure}

Consider the vectorial sentence $\nu^\omega_1 (x_1,x_2) . (\diamond{\letterB}x_2, \diamond{\letterA}x_2)$.
This describes the property {\em there are arbitrarily long paths with labels in $\letterB\letterA^*$}, and so it is true in all points $\point_i$ and false in all points $\altpoint_i$. The following result immediately implies that this property cannot be defined in the scalar countdown calculus:

\begin{theorem}\label{Thm-Vectorial-Vs-Scalar}
  For every scalar sentence $\phi$, there exists $i<\omega$ s.t. $
    \altpoint_i\in\semantics{\phi} \iff \point_i\in\semantics{\phi}$.
\end{theorem}

\begin{proof}
The heart of the proof is Proposition \ref{Prop-Scalar=Injectively-Ranked} which says that scalar formulas correspond to injectively ranked automata. In such an automaton, whenever the counter corresponding to rank $r$ is modified, the automaton must be in \emph{the same state}, which allows the players to copy their strategies between different positions of the semantic game. For the details, see Appendix~\ref{app:vectorial-scalar}.
\end{proof}

\section{Strictness of the countdown nesting hierarchy}
A natural question is whether greater \emph{coutdown nesting}, i.e. the maximal nesting of $\mu^\alpha$ and $\nu^\alpha$ operators with $\alpha\neq\infty$, results in more expressive power. We give a positive answer: under mild assumptions, the hierarchy is strict. From now on, focus on the monomodal case (i.e. $|\Actions|=1$) and we assume that the only ordinal used by formulas is $\omega$.\footnote
{This assumption could be replaced with a weaker requirement: there exists a maximal ordinal $\alpha$ that we are allowed to use, and $\alpha$ is additively indecomposable.}

\begin{theorem}\label{Thm-Strictness-Nesting}
  For every $k<\omega$, formulas with countdown nesting $k+1$ have strictly more expressive power than those with nesting at most $k$.
\end{theorem}

In order to prove strictness, it suffices to prove it on a restricted class of models. We will show that the hierarchy is strict already on the class of transitive, linear, well-founded models -- i.e. (up to isomorphism) ordinals. 

More specifically, an ordinal $\kappa\in\Ord$ can be seen as a model with $\alpha \to \beta$ iff $\alpha>\beta$.
  Since $\kappa$ is an induced submodel of $\kappa'$ whenever $\kappa\leq\kappa'$, we can consider a single ordinal model with $\kappa$ big enough. For our purposes, the first uncountable ordinal $\omega_1$ is be sufficient.

  We call a subset $S \subset\omega_1$ \emph{stable above $\alpha$} if either $[\alpha,\omega_1)\subset S$ or $[\alpha,\omega_1)\cap S=\emptyset$. A \emph{stabilization point} of a valuation $\val:\Var\to\powerset{\omega_1}$ is the least $\alpha\leq\omega_1$ such that interpretations of all the variables are stable above $\alpha$.

Observe that the set $[\omega^k,\omega_1)\subset[0,\omega_1)$ can be defined by the following sentence with countdown nesting $k$:
\begin{align}\label{Eq-Language-Omega^k}
\textstyle
  [\omega^k,\omega_1) = \semantics{\nu^\omega x_1 ... \nu^\omega x_k.\Diamond(\bigwedge_{i\leq k}x_i)}.
\end{align}

Indeed, the semantic game can be decomposed into two alternating steps: (i) $\adam$ chooses a tuple of finite ordinals $(\alpha_1,...,\alpha_k) \in \omega^k$ and (ii) $\eve$ responds with a successor in the model. Since at each step $\adam$ has to pick a lexicographically smaller tuple (and he starts by picking any tuple) it is easy to see that he wins iff the initial point is at least $\omega^k$. We will show that for all $k>0$, countdown nesting $k$ is \emph{necessary} to define this language. The proof relies on the following lemma.

\begin{lemma}\label{Lem-Strictness}
  For every $k<\omega$ and a formula $\phi$ with countdown nesting $k$, there exists an ordinal $\alpha_\phi< \omega^{k+1}$ such that $\phi$ stabilizes $\alpha_\phi$ above the valuation, i.e. for every valuation $\val$ stabilizing at $\beta$, $\semantics{\phi}^\val$ is stable above $\beta+\alpha_\phi$.
\end{lemma}
\begin{proof}
See Appendix~\ref{app:strictness}.
\end{proof}

From this the theorem follows immediately, as the sentence $\phi$ has no free variables and thus it stabilizes at $\alpha_\phi<\omega^{k+1}$ regardless of the valuation. 

\section{Decidability issues}
We briefly discuss decidability issues in the countdown $\mu$-calculus. Note that in a finite model every monotone map reaches its fixpoints in finitely many steps. Hence, if we replace every $\eta^\alpha$ in $\phi$ with $\eta^\infty$ and denote the resulting formula by $\widehat{\phi}$, then \emph{in every finite model} $\semantics{\phi}= \semantics{\widehat{\phi}}$. It immediately follows that:

\begin{proposition}
  The model checking problem for the $\muMLC$, i.e. the problem: ``Given $\phi\in\muMLC$ and a point $\point$ in a (finite) model $\M$, does $\point\models\phi$?'' is decidable.
\end{proposition}

Note that as a corollary we get that deciding the winner of a given (finite) countdown game $\semanticGame$ is also decidable, as set of positions where $\eve$ wins can be easily defined in $\muMLC$.

A more interesting problem is \emph{satisfiability}: ``Given $\phi\in\muMLC$, is there a model $\M$ and a point $\point$ s.t. $\point\models\phi$?''.

\begin{proposition}
  A formula $\phi\in\muMLC$ has \emph{positive countdown} if it does not use $\nu^\alpha$ with $\alpha\neq\infty$. The satisfiability problem is decidable for such formulas.
\end{proposition}
\begin{proof}
  Observe that for $\phi$ with positive countdown, in every model we have $\semantics{\phi}\subset\semantics{\widehat{\phi}}$. Hence, if $\phi$ is satisfiable, then so is $\widehat{\phi}$ -- but since $\muML$ has a finite model property, this means that $\widehat{\phi}$ has a \emph{finite} model, where $\widehat{\phi}$ and $\phi$ are equivalent. Thus, $\phi$ is satisfiable iff $\widehat{\phi}$ is, and the problem reduces to $\muML$ satisfiability.
\end{proof}
Dualizing the above we get that the \emph{validity} problem is decidable for formulas with \emph{negative countdown}, i.e. with $\alpha=\infty$ for every $\mu^\alpha$.

\bibliographystyle{plainurl}
\bibliography{bib}

\newpage
\appendix

\input{appendix-games}

\input{appendix-formtoaut}
\input{appendix-auttoform}
\input{appendix-guarded}
\input{appendix-vectorial-scalar}
\input{appendix-strictness}

\end{document}

%% file: appendix-games.tex
\section{Technical observations about countdown games}\label{app:games}

The following basic results about countdown games will be useful in subsequent sections.

As we mentioned, it is beneficial for each player to keep her/his counter as big as possible. More precisely, given a countdown game, define a partial order $\worse_\exists$ on its configurations: $\pconf{v,\ctr}\worse_\exists\pconf{v,\ctr'}$ and $\cconf{v,\ctr}\worse_\exists\cconf{v,\ctr'}$ if and only if $\ctr(r)\leq\ctr'(r)$ for all $r\in\nonstandardRanks_\exists$ and $\ctr(r)\geq\ctr'(r)$ for all $r\in\nonstandardRanks_\forall$. A partial order $\worse_\forall$ is defined analogously, and it is clearly the inverse of $\worse_\exists$, that is $\gamma\worse_\forall\delta$ if and only if $\delta\worse_\exists\gamma$ for all configurations $\gamma,\delta$. Both orders extend to plays seen as sequences of configuration and compared pointwise. It easily follows from the definition that if $\gamma\worse_\exists\delta$ and $\eve$ has a move from $\gamma$ to a configuration $\gamma'$, then she has a move from $\delta$ to a $\delta'$ such that $\gamma'\worse_\exists\delta'$. A symmetric statement holds for $\adam$. As a result, if $\eve$ has a winning strategy from $\gamma$ and $\gamma\worse_\exists\delta$ then she has a winning strategy from $\delta$; analogously for $\adam$.

Another easy observation is that if the countdown starts from limit ordinals, one may always choose values greater by a finite $k$ than the ones given by some fixed strategy. More specifically,
%
for a counter valuation $\ctr$ and a number $k<\omega$, define a valuation $\ctr+_\exists k$ by:
\begin{itemize}
\item $(\ctr+_\exists k)(r)=\text{min}(\ctr(r)+k,\bounds(r))$ for $r\in\nonstandardRanks_\exists$, and
\item $(\ctr+_\exists k)(r)=\ctr(r)$ for $r\in\nonstandardRanks_\forall$.
\end{itemize}
For a play $\pi$, let $\pi+_\exists k$ be the play pointwise equal to $\pi$ on positions, and replacing every counter valuation $\ctr$ with $\ctr+_\exists k$. If all the initial values $\bounds$ are limit ordinals, then for every strategy $\sigma$ there exist a strategy \emph{$k$-above $\sigma$}, denoted $\sigma+_\exists k$, such that $\pi$ is a $(\sigma+_\exists k)$-play iff $\pi=\pi'+_\exists k$ for some $\sigma$-play $\pi'$. Moreover, if $\sigma$ is winning for $\eve$ then so is $\sigma+_\exists k$. An analogous argument works for $\adam$'s winning strategies.

As mentioned, countdown games are not positionally determined. Below we show a much weaker (yet still useful) property: players can win with strategies that do not depend on the counters in finite stages of the game. Consider a countdown game $(V,E,\rank, \bounds)$. For a countdown play:
 \[
   \pi=v_1\ctr_1 v_2...\ctr_{n-1},v_n, \text{\ \ \ \ or \ \ \ \ } \pi=v_1\ctr_1 v_2...\ctr_{n-1},v_n\ctr_n
 \] 
denote by $\positions(\pi)$ the sequence of consecutive positions $v_1...v_n$. Given a phase $\phase$ of the game and a strategy $\sigma$ for player $P$, we say that a partial function $f:V^*\to~V$ \emph{guides} $\sigma$ in $\phase$ if for every $v\in V$ and $\sigma$-plays $\pi,\pi v\in\phase$ such that $v$ is a position chosen by $P$, the value $f(\positions(\pi))$ is defined and equals $v$. We say that $\sigma$ is \emph{counter-independent in $\phase$} or \emph{$\phase$-counter-independent} iff it is guided in $\phase$ by some partial function called the \emph{$\phase$-component} of $\sigma$ and denoted $\sigma^\phase$. Phase $\phase$ is \emph{proper} if membership in $\phase$ does not depend on the counter values, meaning that for plays $\pi,\pi'$ of the same length, $\positions(\pi)=\positions(\pi')$ implies $\pi\in\phase\iff\pi'\in\phase$.

\begin{proposition}\label{Prop-Finite-CtrIndep}
Take a countdown game $\game=(V,E,\rank, \bounds)$ and a proper phase $\phase$ of $\game$. Assume that the set $\positions[\phase]=\{\positions(\pi)\ |\ \pi\in\phase\}$ is finite. If $\eve$ wins from configuration $\gamma_I$, then she wins with a strategy that is counter-independent in $\phase$. 
\end{proposition}

\begin{proof}
The assumption on $\positions[\phase]$ implies that there exists a finite bound $l_\max$ on the length of plays in $\phase$.
Consider a winning strategy $\sigma$ for $\eve$.
We show by induction on $0\leq l\leq l_\max$ that:
  \begin{align*}
    \parbox[t]{12.5cm}{
      For every $\sigma$-play $\pi$ of length $|\pi|=l_\max-l$,
      there exists a winning strategy $\sigma_\pi$ for $\game,\gamma_I$ that is counter-independent in the subphase $\phase_\pi$ of $\phase$ and equal to $\sigma$ on plays without a prefix from $\phase_\pi$.
    }
  \end{align*}
  Once we prove the claim for $l=l_
  \max$, we obtain a strategy $\sigma_{\epsilon}$ counter-independent in $\phase_{\epsilon}=\phase$, as desired.

The base case is $l=0$ where there is nothing to prove, as $|\pi|=l_\max$ implies that either $\phase_\pi=\{\pi\}$ if $\pi\in\phase$ or $\phase_\pi=\emptyset$ otherwise. In both cases $\sigma$ is trivially guided in $\phase_\pi$ by a partial function undefined on every argument.

For the inductive step, assume that the claim is true for $l$ and for every $\sigma$-play $\pi$ with $|\pi|=l_\max-l$ denote the $\phase_\pi$-component of $\sigma_\pi$ by $\sigma^{\phase_\pi}$. Given a $\sigma$-play $\pi$ with $|\pi|=l_\max-l-1$, there are three cases to consider:
\begin{itemize}
\item After $\pi$ it is $\eve$ who makes a move. Since $\pi$ is a $\sigma$-play, $\sigma$ provides a move $z=\sigma(\pi)$. Since $\pi z$ is also a $\sigma$-play and $|\pi z|=l_\max-l$, by induction hypothesis there exists a winning $\sigma_{\pi z}$ that is counter-independent in $\phase_{\pi z}$.
$\eve$ can therefore win with the strategy: 
\[
  \sigma_\pi(\rho) =
  \begin{cases}
    \sigma_{\pi z}(\rho) & \text{if $\pi z$ is a prefix of $\rho$},\\
    \sigma(\rho) & \text{otherwise.}
  \end{cases}
\]
Unless $\pi,\pi z\in\phase$ and $z\in V$, the strategy $\sigma_{\pi}$ is guided by $\sigma^{\phase_{\pi}}=\sigma^{\phase_{\pi z}}$ in $\phase_\pi$ and otherwise it is guided by:
\[
  \sigma^{\phase_\pi}(\overline{v}) = 
  \begin{cases}
    z & \text{if $\overline{v}=\positions(\pi)$},\\
    \sigma^{\phase_{\pi z}}(\overline{v}) & \text{otherwise.}
  \end{cases}
\]

\item After $\pi$ $\adam$ chooses a position $v$ from a set $W\subset V$. For every such $v$, $\pi v$ is a $\sigma$-play, $|\pi v|=l_\max-l$ and hence induction hypothesis provides $\sigma_{\pi v}$ guided by $\sigma^{\phase_{\pi v}}$ in $\phase_{\pi v}$. 
We combine strategies for all the possible choices from $W$:
\[
  \sigma_\pi(\rho) =
  \begin{cases}
    \sigma_{\pi v}(\rho) & \text{if $\pi v$ is a prefix of $\rho$},\\
    \sigma(\rho) & \text{otherwise.}
  \end{cases}
\]
Such $\sigma_\pi$ is guided in $\phase_\pi$ by:
\[
  \sigma^{\phase_\pi}(\overline{v}) = 
  \begin{cases}
    \sigma^{\phase_{\pi v}}(\overline{v}) & \text{$\overline{v}$ has $\positions(\pi v)$ as a prefix,}\\
    \text{undefined} & \text{otherwise.}
  \end{cases}
\]

\item After $\pi$ $\adam$ updates the current counters $\ctr$ to $\ctr'$. The only interesting case is when the current rank $r$ is nonstandard and so $\ctr'$ is given by a choice of an ordinal $\alpha<\ctr(r)$ (the case with standard $r$ is similar to the first one). Denote such $\ctr'$ by $\ctr_\alpha$ and $\pi\ctr_\alpha$ by $\pi_\alpha$. For every $\alpha<\ctr(r)$ the play $\pi_\alpha$ is consistent with $\sigma$ and $|\pi_\alpha|=l_\max-l$, so induction hypothesis gives us $\sigma_{\pi_\alpha}$ guided by $\sigma^{\phase_{\pi_\alpha}}$ in $\phase_{\pi_\alpha}$.

Observe that for plays $\pi_\alpha,\pi_\beta$ leading to configurations $\gamma_\alpha$ and $\gamma_\beta$, respectively, we have $\gamma_\alpha\succcurlyeq_\exists \gamma_\beta$ whenever $\alpha<\beta$. It follows that if after $\pi$ $\adam$ chooses $\alpha$, $\eve$ may as well continue as if he picked $\beta$. More precisely, she may play maintaining the invariant that for the current play $\pi_\alpha\xi$ there exists a $\sigma_\beta$-play $\pi_\beta\xi'$ with $\pi_\beta\xi'\worse_\exists\pi_\alpha\xi$. Denote such strategy by $\sigma_{[\beta/\alpha]}$.

Importantly, if $\alpha\leq\beta$ and $\sigma^{\phase_{\pi_\beta}}$ guides $\sigma_{\pi_\beta}$ in $\phase_{\pi_\beta}$, then it also guides $\sigma_{[\beta/\alpha]}$ in $\phase_{\pi_\alpha}$.
This is because whenever $\sigma_{[\beta/\alpha]}$-plays $\pi_\alpha\xi$ and $\pi_\alpha\xi v$ belong to $\phase_{\pi_\alpha}$ and $v$ is chosen by $\eve$, there exists $\pi_\beta\xi'\worse_\exists\pi_\alpha\xi$ such that $\pi_\beta\xi'$ and $\pi_\beta\xi'v$ are $\sigma_{\pi_\beta}$-plays. Since $\pi_\beta\xi'\worse_\exists\pi_\alpha\xi$ implies $\positions(\pi_\alpha\xi)=\positions(\pi_\beta\xi')$, by properness of $\phase$ both $\pi_\beta\xi'$ and $\pi_\beta\xi'v$ belong to $\phase_{\pi_\beta}$. Hence, $\sigma^{\phase_{\pi_\beta}}(\positions(\pi_\alpha\xi))=\sigma^{\phase_{\pi_\beta}}(\positions(\pi_\beta\xi'))=v$, as desired.

There are two cases to consider, depending on whether $\ctr(r)$ is a limit ordinal or not. If it is a successor ordinal then there is a maximal $\alpha$ that can be chosen by $\adam$. In that case, $\eve$ uses the strategy:
\[
  \sigma_\pi(\rho)=
  \begin{cases}
    \sigma_{[\alpha/\beta]}(\rho) & \text{if $\pi_\beta$ is a prefix of $\rho$,}\\
    \sigma(\rho) & \text{otherwise,}
  \end{cases}
\]
guided in $\phase_\pi$ by $\sigma^{\phase_\pi}=\sigma^{\phase_{\pi_\alpha}}$.

On the other hand, if $\ctr(r)$ is a limit ordinal then there is no maximal $\alpha$ that $\adam$ can choose, and for each of his choices $\eve$ might have used a different $\sigma_{\pi_\alpha}$.
However, by assumption the set of positions that appear in $\phase$ is finite. As a consequence, there are only finitely many possible partial functions guiding $\sigma_{\pi_\alpha}$ in $\phase_{\pi_\alpha}$ and we may find $\sigma^{\phase_{\pi}}$ such that $\sigma_{\pi_\alpha}$ is guided in $\phase_{\pi_\alpha}$ by $\sigma^{\phase_\pi}$ for arbitrarily big $\alpha<\ctr(r)$. Define:
\[
  \sigma_\pi(\rho)=
  \begin{cases}
    \sigma_{[\alpha'/\alpha]}(\rho) & \text{if $\pi_{\alpha}$ is a prefix of $\rho$,}\\
    \sigma(\rho) & \text{otherwise,}
  \end{cases}
\]
where $\alpha'\geq\alpha$ is the least number greater than $\alpha$ with $\sigma^{\phase_{\pi_\alpha}}=\sigma^{\phase_\pi}$. By design, $\sigma_\pi$ is guided by $\sigma^{\phase_\pi}$ in $\phase_\pi$.
\end{itemize}
\end{proof}

Sometimes we will want to decompose games into smaller fragments. Given a game $\game=(V,E,\rank,\bounds)$ and a subset $\exitPos\subset V$, the \emph{partial game} $\game|\exitPos$ is played the same as $\game$, except that upon reaching a configuration with a position in $\exitPos$ (this is always a countdown configuration, as they are the first to be reached in any given position) the play ends with a draw, meaning that no player looses or wins. For a non-loosing $\eve$'s strategy $\sigma$ for $\game|\exitPos,v$, the set of its \emph{exit positions}, denoted $\exit(\sigma)\subset\exitPos$, consists of all the positions $u\in\exitPos$ such that some $\sigma$-play ends in a countdown configuration in $u$.

Given another $\game'=(V',E',\rank',\bounds')$ and a subset of positions $\exitPos\subset V\cap V'$, we say that $\game,v$ is \emph{exit-equivalent} to $\game',v'$ with respect to $\exitPos$, denoted $\game,v\exitEquiv{\exitPos}\game',v'$, if for every non-loosing $\eve$'s strategy $\sigma$ for $\game|\exitPos,v$ she has a non-loosing $\sigma'$ for $\game'|\exitPos,v'$ s.t. $\exit(\sigma)=\exit(\sigma')$, and symmetrically for every $\sigma'$ for $\game'|\exitPos,v'$. Exit-equivalence $\exitEquiv{\exitPos}$ is an equivalence relation between games having $Z$ as a subset of its positions.

\begin{lemma}[Decomposition Lemma]\label{Lem-Decomposition}
  Consider games $\game$ and $\game'$ as above such that the most important ranks $r=\max(\rangeRank)$ and $r'=\max(\rangeRank')$ have the same type (meaning $r\in\rangeRank_\exists$ iff $r'\in\rangeRank'_\exists$ and either both are standard or $\bounds(r)=\bounds'(r')$). Assume disjoint subsets $\exitPos_\exit,Z_\max\subset V\cap V'$ with $\exitPos_\max=\rank^{-1}(r)=\rank'^{-1}(r')$. For every $v_I\in V$ and $v'_I\in V'$:
  \begin{itemize}
    \item $\game,v_I \exitEquiv{\exitPos_\exit\cup \exitPos_\max} \game',v'_I$ and:
    \item $ \game,v \exitEquiv{\exitPos_\exit\cup \exitPos_\max} \game',v$  for all $v\in \exitPos_\max$,
  \end{itemize}
  implies $\game,v_I \exitEquiv{\exitPos_\exit} \game',v'_I$.
\end{lemma}

\begin{proof}
  Playing $\game|\exitPos_\exit,v_I$ can be decomposed into a sequence of alternating choices starting in the initial configuration $\gamma_0=\pconf{v_I,\bounds}$:
\begin{enumerate}
  \item\label{step1} In a configuration $\gamma_i=\pconf{v_i,\ctr_i}$ $\eve$ declares a fragment $\sigma_i$ of her strategy that determines her moves until some countdown configuration in $\exitPos_\exit\cup \exitPos_\max$ is reached;
  \item $\adam$ chooses a finite $\sigma_i$-play $\pi_i$ leading to a configuration $\cconf{v_{i+1},\ctr}$ for some $v_{i+1}\in\exitPos_\exit\cup\exitPos_\max$;
  \item If $v_{i+1}\in\exitPos_\exit$ the game ends in a draw, otherwise the owner of $r=\rank(v_{i+1})$ updates $\ctr$ to some $\ctr_{i+1}$, and the play proceeds from $\gamma_{i+1}=\pconf{v_{i+1},\ctr_{i+1}}$ as in step~\eqref{step1}~above.
\end{enumerate}

For every configuration $c_i$ as above, the fragments of strategies $\sigma_i$ chosen by $\eve$ in (1) can be identified with strategies for $\game|\exitPos_\exit\cup\exitPos_\max,\pconf{v_i,\ctr_i}$. Since $\game|\exitPos_\exit\cup\exitPos_\max$ always stops before any position in $\exitPos_\max$ is reached, the counter value for $r$ is irrelevant \emph{inside} $\game|\exitPos_\exit\cup\exitPos_\max$:
\[
  \game,\pconf{v_i,\ctr_i}\exitEquiv{\exitPos_\exit\cup\exitPos_\max} \game, \pconf{v_i,\ctr_i[r\mapsto \beta]}
\]
for every $\beta\in\Ord$. Moreover, ranks greater than $r$ are never reached before the game stops and whenever the game enters a position $v\in\exitPos_\max$, counters for all ranks smaller than $r$ are reset back to their initial values from $\bounds$, so each counter assignment $\ctr_i$ has to be of the form $\bounds[r\mapsto \beta]$. Thus, the fragments of strategies in (1) are the same as strategies for $\game|\exitPos_\exit\cup\exitPos_\max,\pconf{v_i,\bounds}$. Finally, since countdown games are configurationally determined, we may assume that in step (2) $\adam$ picks only a position $v_{i+1}\in\exitPos_\max$ (instead of an entire play $\pi_i$ ending with a countdown configuration  in $v_{i+1}$).

Consider a game $\widehat{\game}$ starting in $v_0=v_I$ and consisting of three alternating phases:
\begin{enumerate}
  \item From $v_i\in V$, $\eve$ picks a non-loosing strategy $\sigma_i$ for $\game|\exitPos_\exit\cup\exitPos_\max,v_i$;
  \item $\adam$ chooses an exit position $v_{i+1}\in\exit(\sigma_i)$, that is, a position in $\exitPos_\max$  reachable by $\sigma_i$ ;
  \item if $r$ is nonstandard, its owner decrements the corresponding counter, otherwise nothing changes, and in both cases we proceed from $v_{i+1}$.
\end{enumerate}

Formally, the game $\widehat{\game}$ is a countdown game with positions $\widehat{V}_\exists=V$ and $\widehat{V}_\forall=\bigcup_{v\in V}S_v$, edges $\widehat{E}=\{(v,\sigma)\ |\ v\in V, \sigma\in S_v\}\cup\{(\sigma,v)\ |\ v\in \exit(\sigma)\}$ where:
\[
  S_v=\{\sigma\ |\ \text{$\sigma$ is a non-loosing strategy for $\game|\exitPos_\exit\cup\exitPos_\max, v$}\}
\]
and $\widehat{\rank}(v)=r$ and $\widehat{\rank}(\sigma)=0$ with $\widehat{\rangeRank}=\{0\preceq r\}$ where $r$ has the same type as in $\game$ and $0$ is an irrelevant, standard rank.
It follows from the discussion above that:
\[
  \game,v \exitEquiv{\exitPos_\exit}\widehat{\game},v
\]  
for every $v\in V$. On the other hand, the assumptions of the lemma imply that the relation:
\[
  \{(v_I,v'_I)\}\cup\{(v,v)\ |\ v\in\exitPos_\max\}\cup\{(\sigma,\sigma')\in \widehat{V}_\forall\times\widehat{V'}_\forall\ |\ \exit(\sigma)=\exit(\sigma')\}
\]
is a bisimulation between the arenas of $\widehat{\game}|\exitPos_\exit$ and $\widehat{\game'}|\exitPos_\exit$ that preserves order and type of ranks and ownership of positions. It follows that:
\[
  \game, v_I \exitEquiv{Z_\exit} \widehat{\game}, v_I \exitEquiv{Z_\exit} \widehat{\game'}, v'_I \exitEquiv{\exitPos_\exit} \game', v'_I
\]
which proves the lemma.
\end{proof}

%% file: appendix-formtoaut.tex
 \section{Proof of Theorem~\ref{Thm-Adequacy-Countdown}}\label{app:formtoaut}
 
 Unfolding the definition of $\semantics{\A_\phi}^\val$ from Definition~\ref{def:semantic-game}, we prove that
 \begin{align}\label{eq:IH1}
 	 \point \in \semantics{\phi}^\val \iff \eve \text{ wins } \semanticGame^\val(\phi) \text{ from } \pconf{(\point,\phi),\bounds}
 \end{align}
by induction on $\phi$. The only interesting case is when $\phi = \mu^\alpha_i \overline{x}.\overline{\phi}$ for some $\overline{x} = \langle x_1, ..., x_n\rangle$, $\overline{\phi} = \langle \phi_1, ..., \phi_n \rangle$ and $\alpha \in \Ord_\infty$ (the case of $\nu^\alpha$ in place of $\mu^\alpha$ is symmetric). If $\alpha = \infty$ (i.e. $\mu^\alpha$ is just a usual fixpoint operator), then the proof is essentially the same as for the classical $\mu$-calculus~\cite{Ven20}. For $\alpha \in \Ord$, a different argument is needed. We prove~\eqref{eq:IH1} by induction on $\alpha$:
\begin{align}\label{eq:IH2}
  \point \in \semantics{\mu^\alpha_i \overline{x}. \overline{\phi}}^\val \iff \eve \text{ wins } \semanticGame^\val(\mu^\alpha_i \overline{x}. \overline{\phi}) \text{ from } \pconf{(\point,\mu^\alpha_i \overline{x}. \overline{\phi}),\bounds}.
\end{align}
To this end, denote $ H_j^\beta = \semantics{\mu^\beta_j \overline{x}. \overline{\phi}}^\val$  for $j \leq n$ and $\beta \in \Ord$. By Definition~\ref{def:vec-semantics} we have:
  \[
    H_j^\alpha = \bigcup_{\beta<\alpha}\semantics{\phi_j}^{\val_\beta}, \qquad\text{where } \val_\beta = \val[x_1 \mapsto H_1^\beta, ..., x_n \mapsto H_n^\beta].
  \]
By the induction hypothesis~\eqref{eq:IH1} applied to $\phi_j$, $\point \in H_j^\alpha$ if and only if there exists a $\beta<\alpha$ such that $\eve$ wins the game $\semanticGame^{\val_\beta}(\phi_j)$ from the position $\pconf{(\point,\phi_j),\bounds}$. 

Consider now the game $\semanticGame^\val(\mu^\alpha_i \overline{x}.\overline{\phi})$ and its initial positional configuration $\pconf{(\point,\mu^\alpha_i \overline{x}. \overline{\phi}),\bounds}$. The first move from this configuration is deterministic, to the countdown configuration $\cconf{(\point,\phi_i),\bounds}$. The next move in the game is made by $\eve$, as she is the owner of $r=\rank(\phi_i)$. Note that $\bounds(r)=\alpha$, and $r$ is the highest rank for which $\bounds$ is defined. Therefore $\eve$ chooses some $\beta<\alpha$ and moves to the configuration $\pconf{(\point,\phi_i),\bounds[r\mapsto\beta]}$.

The game $\semanticGame^\val(\mu^\alpha_i \overline{x}.\overline{\phi})$, played from this configuration, does not differ from $\semanticGame^{\val_\beta}(\phi_i)$ played from $\pconf{(\point,\phi_i),\bounds}$, until some variable $x_j$ is reached. If this happens, in the former game we continue in a game isomorphic to $\semanticGame^\val(\mu^\beta_j \overline{x}.\overline{\phi})$. In the latter game, $\eve$ wins if and only if the current point $\altpoint$ belongs to $H_j^\beta$. Since $\beta<\alpha$, by the induction hypothesis~\eqref{eq:IH2} these two conditions are equivalent. This finishes the proof.

%% file: appendix-auttoform.tex
 \section{Proof of Theorem~\ref{thm:auttoform}}\label{app:auttoform}
 
   Fix an automaton $\A=(Q,q_I,\delta,\rank,\bounds)$. For clarity of presentation we only consider the case when $\A$ has no free variables, the general case requires no new ideas. Without losing generality assume that the highest rank $r_\max$ is not assigned to any state and every other rank is assigned to at least one state.
  
We construct, by induction on $r\in \rangeRank$, a formula $\psi_{r,q}$ over the set $Q$ treated as formal variables, with $\rank(q)\geq r$ whenever $q$ occurs free and $\rank(q)<r$ if it occurs bound, and such that for every $\point\in\M$:
  \begin{align}\label{eq:A2F induction hypothesis}
    \semanticGame(\A),(\point,q) \exitEquiv{\exitPos_r} \semanticGame(\psi_{r,q}),(\point,\psi_{r,q})
  \end{align}
  where:
  \[
    \exitPos_r=\{(\altpoint,q)\in M\times Q\ |\ r\leq\rank(q)\}. 
  \]
  Note that although formally $\psi_{r,q}$ may contain free variables, the game $\semanticGame^\val(\psi_{r,q})|\exitPos_r,(\point,\psi_{r,q})$ always stops before any such variable is reached, so we ignore the valuation $\val$ and write $\semanticGame(\psi_{r,q})$.
  
  Given~\eqref{eq:A2F induction hypothesis}, since no state in $\A$ has the highest rank $r_\max$, the set $\exitPos_{r_\max}$ is empty and so the games $\semanticGame(\A),(\point,q_I)$ and $\semanticGame(\psi_{r_\max,q_I}),(\point,\psi_{r_\max,q_I})$ are equivalent, which will prove the theorem.
  
  Denote the lowest rank by $0$. The set $\exitPos_0$ contains all the positions of $\semanticGame(\A)$, meaning that $\semanticGame(\A)|\exitPos_0$ stops immediately after the first move. Thus for the base case of~\eqref{eq:A2F induction hypothesis} it is enough to put:
  \begin{itemize}
    \item if $\delta(s)=(\action,p)$:
    \[
      \psi_{0,s} =
       \begin{cases}
        \diamond{\action}p & \text{ if $q$ belongs to $\eve$}\\
        \boxmodal{\action}p & \text{ if $q$ belongs to $\adam$}
       \end{cases}
    \]
    \item if $\delta(s)\subset Q$:
     \[
      \psi_{0,s} =
       \begin{cases}
        \bigvee\delta(s) & \text{ if $q$ belongs to $\eve$}\\
        \bigwedge\delta(s) & \text{ if $q$ belongs to $\adam$}.
       \end{cases}
     \]
  \end{itemize}
 
  For the inductive step, assuming~\eqref{eq:A2F induction hypothesis} for $r$, we will prove it for the next rank, denoted $r+1$. Let $q_1,...,q_d$ be all states in $Q$ with rank $r$. For every $q_i$ define the vectorial formula:
  \[
    \theta_i=\eta^\alpha_{q_i}(q_1,...,q_d).(\psi_{r,q_1},...,\psi_{r,q_d})
  \]
  with $\alpha=\bounds(r)$ and $\eta=\mu$ if $r$ belongs to $\eve$ and $\eta=\nu$ if $r$ belongs to $\adam$. Then put:
  \begin{align}\label{eq:psirp1q}
    \psi_{r+1,q}= \psi_{r,q}[q_1\mapsto\theta_1, ..., q_d\mapsto\theta_d]
  \end{align}
  for every $q\in Q$. We need to prove that:
  \begin{align}\label{eq:A2F induction goal}
    \semanticGame(\A),(\point,q) \exitEquiv{\exitPos_{r+1}} \semanticGame(\psi_{r+1,q}) ,(\point,\psi_{r+1,q})
  \end{align}
  for all $\point\in\M$.

  The arena $V'=M\times\SubFor(\psi_{r+1,q})$ of the game $\semanticGame(\psi_{r+1,q})=(V',E',\rank',\bounds')$ decomposes into $V^\textit{II}=\bigcup_{1\leq i \leq d}V^\textit{II}_i$ for $V^\textit{II}_i=M\times\SubFor(\theta_i)$ and $V^\textit{I}=V'-V^\textit{II}$. Once a play enters $V^\textit{II}_i$ it stays there forever, since a move from $V^\textit{II}_i$ to $V^\textit{I}$ would only be possible if there was a variable free in $\theta_i$ but bound in its proper superformula (and hence also bound in $\psi_{r,q}$). However, if $p$ is bound in $\psi_{r,q}$ then $\rank(p)<r$, whereas $p$ can be free in $\theta_i$ only if $r\leq\rank(p)$. This implies that in $\psi_{r+1,q}$ the only formulas reachable from $\theta_i$ are its strict subformulas. Therefore, putting:
  \[
    \theta = \eta_y^\alpha(y,q_1,...,q_d).(\psi_{r,q},\psi_{r,q_1},...,\psi_{r,q_d})
  \]
  with $y$ a fresh variable we obtain:
  \[
    \semanticGame(\psi_{r+1,q}) ,(\point,\psi_{r+1,q}) \exitEquiv{\exitPos_{r+1}} \semanticGame(\theta) ,(\point,\psi_{r,q})
  \]
  because $V^\textit{I}$ corresponds to $M\times\SubFor(\psi_{r,q})$ and $V^\textit{II}$ to $\bigcup_{i\leq d}M\times\SubFor(\psi_{r,q_i})$ with freshness of $y$ guaranteeing that there is no return from the second part to the first one.

  Denote the rank of $\psi_{r,q_1},...,\psi_{r,q_d}$ in $\semanticGame(\theta)$ by $r'$ and recall that $q_1,...,q_d$ all have rank 0. Consider the game $\widetilde{\semanticGame(\theta)}$ that is the same as $\semanticGame(\theta)$ except for the ranking function that swaps $r'$ and 0, i.e. $\psi_{r,q}$ and each $\psi_{r,q_i}$ have rank $0$ and each $q_i$ has rank $r'$. Since in $\semanticGame(\theta)$: (i) a move has $(\altpoint,\psi_{r,q_i})$ as a target iff it has $(\altpoint,q_i)$ as a source and (ii) no nonempty play starting at $(\point,\psi_{r,q})$ reaches any $(\altpoint,\psi_{r,q})$, we have:
  \begin{align*}
    \semanticGame(\theta) ,(\point,\psi_{r,q}) \exitEquiv{\exitPos_{r+1}} \widetilde{\semanticGame(\theta)} ,(\point,\psi_{r,q})
  \end{align*}
  and so for \eqref{eq:A2F induction goal} it remains to prove:
  \begin{align}\label{eq:A2F induction goal 2}
    \semanticGame(\A),(\point,q) \exitEquiv{\exitPos_{r+1}} \widetilde{\semanticGame(\theta)} ,(\point,\psi_{r,q}).
  \end{align}
  Note that:
  \[
    \exitPos_{r}=\exitPos_{r+1}\cup Y_r \text{\ \ \ \ \ \ \ with \ \ \ \ \ \ \ } Y_r=M\times\{q_1,...,q_d\}.
  \]
  and $Y_r$ is precisely the set of positions with rank $r$ and $r'$ in $\semanticGame(\A)$ and $\widetilde{\semanticGame(\theta)}$, respectively. Since the ranks $r$ and $r'$ have the same type and are the most important in both games, by Lemma~\ref{Lem-Decomposition} to prove~\eqref{eq:A2F induction goal 2} it  is enough to prove that:
  \begin{enumerate}
    \item $\semanticGame(\A),(\point,q) \exitEquiv{\exitPos_{r}} \widetilde{\semanticGame(\theta)},(\point,\psi_{r,q})$, and
    \item $ \semanticGame(\A), (\altpoint,q_i) \exitEquiv{\exitPos_{r}} \widetilde{\semanticGame(\theta)},(\altpoint,q_i)$
    for all $(\altpoint,q_i)\in Y_r$.
  \end{enumerate}

  There are two cases to consider:
  \begin{enumerate}
    \item For $(\point,q)$ and $(\point,\psi_{r+1,q})$:
    \begin{align*}
      \semanticGame(\A),(\point,q) &\exitEquiv{\exitPos_r} \semanticGame(\psi_{r,q}),(\point,\psi_{r,q})\\
      &\exitEquiv{\exitPos_r} \semanticGame(\theta),(\point,\psi_{r,q})\\
      &\exitEquiv{\exitPos_r} \widetilde{\semanticGame(\theta)},(\point,\psi_{r,q}).
    \end{align*}
    The first equivalence is given by the induction hypothesis.
    The second one is true because the partial games $\semanticGame(\psi_{r,q})|\exitPos_r,(\point,\psi_{r,q})$ and $\semanticGame(\theta)|\exitPos_r,(\point,\psi_{r,q})$ are isomorphic.
    The third one follows from the observation that the difference between $\semanticGame(\theta)$ and $\widetilde{\semanticGame(\theta)}$ is only in ranks of positions $(\altpoint,\psi_{r,q})$ and $(\altpoint,\psi_{r,q_i}),(\altpoint,q_i)$ for $i\leq d$ and $\altpoint\in M$, but these positions cannot be reached by a nonempty play from $(\point,\psi_{r,q})$ before the game stops, i.e. without passing through $\exitPos_r$.

    \item For $(\altpoint,q_i)\in Y_r$:
    \begin{align*}
      \semanticGame(\A),(\altpoint,q_i) &\exitEquiv{\exitPos_r} \semanticGame(\psi_{r,q_i}),(\altpoint,\psi_{r,q_i})\\
      &\exitEquiv{\exitPos_r} \semanticGame(\theta),(\altpoint,\psi_{r,q_i})\\
      &\exitEquiv{\exitPos_r} \widetilde{\semanticGame(\theta)},(\altpoint,\psi_{r,q_i})\\
      &\exitEquiv{\exitPos_r} \widetilde{\semanticGame(\theta)},(\altpoint,q_i).
    \end{align*}
    The first three equivalences are true for reasons analogous to the previous case. The last one follows from the observation that in $\widetilde{\semanticGame(\theta)}$ the game moves deterministically from $(\altpoint,q_i)$ to $(\altpoint,\psi_{r,q_i})$ and the later position has the least important, standard rank $0$.
  \end{enumerate}

%% file: appendix-guarded.tex
 \section{Guarded formulas}\label{app:guarded}
To demonstrate usefulness of the correspondence between formulas and automata, but also for technical use in further proofs, we shall now show that without loss of generality formulas are \emph{guarded}.
We say that an automaton $\A$ is guarded if it does not contain a loop without modal transitions. A formula $\phi$ is guarded if it is guarded when seen as an automaton $\A_\phi$.

\begin{proposition}\label{Prop-Guardedness}
  Every countdown formula can be transformed into an equivalent guarded one.
\end{proposition}
\begin{proof}
Note that in a countdown game, if a play moves from a position $v$ to itself via a path without visiting ranks higher than $\rank(v)$, then all the counters for lower ranks are reset and those for higher ranks remain unchanged. It follows that the resulting configuration is at least as good for the opponent $P$ of the owner $P'$ of $\rank(v)$ as the one at the previous visit to $v$. Hence, $P$ can repeat the strategy from that moment, and either eventually the game stops looping on $v$ via lower ranks or $P'$ looses. This means that in order to win, $P'$ must have a strategy that avoids such loops, and therefore $P'$ may use that strategy immediately. It follows that we obtain an equivalent game by adding the rule that whenever a play moves from any position $v$ to itself via a path without visiting higher ranks, the owner of $\rank(v)$ immediately looses.

Thanks to this, in any formula we may replace every $\eta^\alpha_k(x_1,...,x_n).(\psi_1,...,\psi_n)$ with:
\[
  \eta^\alpha_{k,0}(x_{i,j})_{i,j\leq n}.(\psi_{i,j})_{i,j\leq n}
\]
where $\psi_{i,j}$ is obtained from $\psi_i$ by replacing 
\begin{itemize}
\item[(i)] every guarded $x_m$ with $x_{m,0}$ and 
\item[(ii)] every other $x_m$ with $\top/\bot$ (resp.) if $j=n$ and $\eta=\nu/\mu$, or with $x_{m,j+1}$ otherwise. 
\end{itemize}f
This way, the number of unguarded unravellings of the $\eta^\alpha$ operator is counted in the index $j$, and the game stops whenever the play passes through the $(n+1)$-st such unravelling (as it implies a repeated visit of a position associated with some $x_m$).
\end{proof}

Let us establish a few more useful facts about countdown automata (and, in light of Section~\ref{sec:formtoaut}, about countdown formulas) that will be useful in Section~\ref{Sec-vectorial-scalar}.

For an automaton $\A$ with states $Q$, a valuation $\val$ and a point $\point_I$ in a model $\M$, the \emph{pre-modal phase} of the game $\semanticGame^\val(\A),(\point_I,q_I)$ consists of all pre-modal plays, i.e. plays with no modal move. All the positions accessible in that phase are of the form $(\point_I,q)$ for $q\in Q$ and if $\A$ is guarded, then no pre-modal play is longer than $|Q|$. Hence, it follows from Proposition \ref{Prop-Finite-CtrIndep} that:

\begin{proposition}\label{Prop-PreModal-CtrIndep}
  In every game $\semanticGame^\val(\A),(\point_I,q_I)$ for a guarded automaton $\A$, the winning player has a pre-modally counter-independent (i.e. counter-independent in the pre-modal phase) winning strategy.
\end{proposition}

Since all the positions appearing in the pre-modal phase only have the initial point on the first coordinate, we can identify pre-modal plays $\pi$ and $\pi'$ starting in $(\point,q)$ and $(\point',q)$ for different $\point\neq\point'$ if $\pi$ equals $\pi'$ after swapping $\point$ and $\point'$. Likewise, we simplify the pre-modal component $\sigma^I:(\{\point\}\times Q)^{< |Q|}\to\{\point\}\times Q$ guiding pre-modal $\sigma$-plays to $\sigma^I:Q^{<|Q|}\to Q$ by skipping the reduntant first coordinate.

\begin{proposition}\label{Prop-Pointwise-Max-Strategy}
  Consider two points $\point_0,\point_1$ in a model $\M$, a valuation $\val$ and a guarded automaton $\A$. Assume that a player $P$ wins the game $\semanticGame^\val(\A)$ from $\point_0$ and $\point_1$ with pre-modally counter-independent strategies $\sigma_0$ and $\sigma_1$, respectively, both guided by the same pre-modal component $\sigma^I$. Then there are winning strategies $\sigma'_0,\sigma'_1$ guided by $\sigma^I$ such that:
  \begin{itemize}
  \item $\sigma'_0$ and $\sigma'_1$ behave the same in the pre-modal phase, up to swapping the points $\point_0$ and $\point_1$, and
  \item for every $(\point_i,\ctr)$ reachable by a $\sigma'_i$-play, there are $(\point_0,\ctr_0)$ and $(\point_1,\ctr_1)$ such that each $(\point_j,\ctr_j)$ is reachable by a $\sigma_j$-play and $\ctr_i\worse_P\ctr$.
  \end{itemize}
\end{proposition}
\begin{proof}
  Starting in $\point_0$ or $\point_1$, $P$ can maintain the invariant that for the play $\pi$ so far, there are $\pi_0$ and $\pi_1$ consistent with $\sigma_0$ and $\sigma_1$ respectively, such that (i) all the three plays are (point-wise) equal on $P$'s choices of positions and on all choices of $P$'s opponent, and (ii) $P$'s choices of counter values in $\pi$ are the maximum of the corresponding choices from $\pi_0$ and $\pi_1$. This way either $P$ wins in the pre-modal phase, or the play reaches a modal move with counter values at least as good for $P$ as after some $\sigma_0$- and $\sigma_1$-plays, respectively. $P$ may then continue from $\point_i$ with the winning strategy $\sigma_i$.
\end{proof}

\begin{proposition}\label{Prop-Convex-Meaning-On-Monotone}
Consider three points $\point_1,\point_2,\point_3$ in a model $\M$ s.t. for every $\action \in \Actions$, the sets $S_1^\action,S_2^\action,S_3^\action$ of their $\action$-successors are monotone, i.e. $S_1^\action\subset S_2^\action\subset S_3^\action$; a valuation $\val$ that does not distinguish $\point_i$ (i.e. $\point_i\in\val(x)\iff~\point_j\in\val(x)$ for all $x\in\Var$); and a guarded automaton $\A$.
  If a player $P$ wins the semantic game $\semanticGame^\val(\A)$ from $\point_1$ and $\point_3$ using strategies $\sigma_1,\sigma_3$ guided by the same pre-modal component $\sigma^I$, then $P$ also wins from $\point_2$ with a strategy $\sigma_2$ guided by $\sigma^I$.
\end{proposition}

\begin{proof}
By Proposition \ref{Prop-Pointwise-Max-Strategy}, we may assume that $\sigma_1$ behaves the same as $\sigma_3$ in the pre-modal phase. Initially $P$ may apply the same strategy from $\point_2$, as the point in the model does not matter, or does not change, in the pre-modal phase. Consider any play consistent with this strategy. If $P$ does not win already in the pre-modal phase, the play reaches a modal move, i.e. a configuration $(\point_2,q)$ with $q\in Q$ such that $\delta(q)=(\action,p)$. If the state $q$ is owned by $P$ then $P$ may continue with $\sigma_1$, and if $q$ is owned by $P$'s opponent then $P$ may continue with $\sigma_3$.
\end{proof}

%% file: appendix-vectorial-scalar.tex
 \section{Proof of Theorem~\ref{Thm-Vectorial-Vs-Scalar}}\label{app:vectorial-scalar}

Observe that since scalar sentences are closed under negation, it is enough to prove that for every scalar $\phi$ there is a $m_\phi$ such that for all $i>m_\phi$:
\begin{align*}
  \tag{$\star$}\label{Impl-Scalar}
  \point_i\in\semantics\phi \implies \altpoint_i\in\semantics\phi
\end{align*}
Moreover, note that every scalar formula can be transformed into an equivalent \emph{guarded} formula that is also scalar, by replacing every unguarded occurrence of a variable bound by $\mu^\alpha$ (or $\nu^\alpha$) by $\bot$ (or $\top$, respectively). Hence, it suffices to prove \eqref{Impl-Scalar} for guarded formulas. For the rest of the proof, we fix a guarded scalar sentence $\phi$ and denote $\game=\semanticGame(\phi)=(V,E,\rank,\bounds)$.

Let us start with an easy fact.

\begin{proposition}\label{Prop-Eventually-Constant}
  There exists some $N<\omega$ such that for all $N \leq i < j$:
  \[
    \point_i\in\semantics\phi \iff \point_j\in\semantics\phi 
  \]
  and if $\eve$ wins the corresponding evaluation games then she does so with pre-modally counter-independent strategies $\sigma_{\point_i}$ and $\sigma_{\point_j}$ with the same pre-modal component $\sigma^I$ that does not depend on $i,j$.
\end{proposition}
\begin{proof}
  Note that the relations $\arrowA$ and $\arrowB$ are monotone, i.e. the bigger $i$, the more $\letterA$- and $\letterB$-successors $\point_i$ has. On the other hand, there are only finitely many possible pre-modal components, so by the pigeonhole principle if $\eve$ wins from $\point_i$  for arbitrarily big $i$, some pre-modal component $\sigma^I$ is used for arbitrarily big $i$. Thus, by Proposition \ref{Prop-Convex-Meaning-On-Monotone} she can use $\sigma^I$ to win for \emph{all $i$ big enough}.
\end{proof}

Towards \eqref{Impl-Scalar}, assume that $\point_i \in\semantics{\phi}$ for all $i$ big enough (otherwise, by Proposition \ref{Prop-Eventually-Constant}, $\phi$ is \emph{false} in $\point_i$ for all $i$ big enough, which trivially implies \eqref{Impl-Scalar}) and denote by $N$ the least number for which Proposition \ref{Prop-Eventually-Constant} holds. Then:

\begin{proposition}\label{Prop-Bounded-A-Diamond}
  Without loss of generality we may assume that for every $i>N$, if a $\sigma_{\point_i}$-play visits a modal position for the first time and it has the shape $(\point_i,\diamond{\letterA}\psi)$, then $\sigma_{\point_i}$ chooses a point $\point_{j}$ for some $j<N$.
\end{proposition}\begin{proof}
  By Proposition \ref{Prop-Eventually-Constant}, $\sigma_{\point_i}$ and $\sigma_{\point_N}$ have the same pre-modal component $\sigma^I$. Therefore, by Proposition \ref{Prop-Pointwise-Max-Strategy}, there is a strategy $\sigma'_{\point_i}$ winning from $\point_i$ guided by the same $\sigma^I$ and only reaching pre-modal configurations at least as good for $\eve$ as the ones reachable by $\sigma_{\point_N}$. Then, whenever $(\point_i,\diamond{\letterA}\psi)$ is reached in a pre-modal $\sigma'_{\point_i}$-play, by the monotonicity of $\arrowA$ and the assumption that $N<i$, $\eve$ may just continue with $\sigma_{\point_N}$. Moreover, since $\sigma_{\point_N}$ is a legitimate strategy, it must pick a point $\point_j$ for some $j<N$, as desired.
\end{proof}

Denote by $\phaseScalar$ the phase of the game $\semanticGame(\phi)$ that consists of plays of the shape $\pi=\xi\rho$ such that the play $\xi$ ends with the first modal move (meaning that every proper prefix of $\xi$ is pre-modal but $\xi$ is not) and $\rho$ does not visit (i) a formula beginning with $\boxmodal{\letterA}, \boxmodal{\letterB}$ or $\diamond{\letterB}$, nor (ii) a formula with a rank that was visited in the pre-modal phase (i.e.~in a proper prefix of $\xi$). Note that the definition allows for empty $\rho$ but not empty $\xi$. The next step is the following claim:

\begin{proposition}\label{Prop-Bound-ArrowA}
 Without loss of generality there exists a finite bound $k_\max<\omega$ such that no $\sigma_{\point_i}$-play $\pi\in\phaseScalar$ contains more than $k_\max$ modal moves.
\end{proposition}
Before proving the above proposition, let us demonstrate how it implies \eqref{Impl-Scalar}. Put:
\[
  m_\phi=k_\max+N+1
\]
where $k_\max$ is the bound from Proposition \ref{Prop-Bound-ArrowA}. We show that $\altpoint_i \in\semantics{\phi}$ for every $i>m_\phi$. To this end, consider the strategy $\sigma_i=\sigma_{\point_i}+_\exists 1$, i.e.~the strategy one above $\sigma_{\point_i}$. 
In the pre-modal phase of the evaluation game from $(\altpoint_i,\phi)$, use $\sigma_i$. Since $\point_i$ and $\altpoint_i$ have the same $\arrowA$-successors, if a play visits a formula beginning with $\diamond{\letterA}$, or $[\letterA]$, $\eve$ may continue with $\sigma_i$ and win. The same is true for $[\letterB]$, as every $\arrowB$-successor of $\altpoint_i$ is also a $\arrowB$-successor of $\point_i$.

The only interesting case is when a play reaches a formula that begins with $\diamond{\letterB}$ and $\sigma_i$ chooses $\point_{j'}$ for some $j'\geq i$ (if $j'<i$ then $\point_{j'}$ is a $\arrowB$-successor of $\altpoint_i$, so $\eve$ may use $\sigma_i$). In this case, $\eve$ may choose $\point_j$ where $j=k_\max+N$ and play maintaining the invariant that for the current play $\pi$, as long as it belongs to $\phaseScalar$, there is a $\sigma_i$-play $\pi'$ in $\phaseScalar$ s.t.:
\begin{enumerate}
  \item\label{Invariant-Item-1} all subformulas and ordinals are the same in $\pi$ and $\pi'$,
  \item\label{Invariant-Item-2} for the last points $\point_j$ and $\point_{j'}$ of $\pi$ and $\pi'$, respectively, we have:
  \[
    k+N \leq j \leq j'  \qquad \text{and}\qquad j<i 
  \]
  where $k$ is the bound on the number of modal moves that can be made after playing $\pi$ (or, equivalently, $\pi'$) before leaving $\phaseScalar$ .
\end{enumerate}

It is straightforward to maintain the invariant on $\epsilon$-transitions and when counter values are updated.

Since $\point_j$ and $\point_{j'}$ have the same $\arrowB$-successors, if after the play $\pi$ the game eventually moves via $\arrowB$ then $\eve$ may continue as with $\sigma_i$, i.e.~using the strategy $\pi\rho\mapsto\sigma_i(\pi'\rho)$. 
Moreover, $j \leq j'$ (guaranteed by item \ref{Invariant-Item-2} of the invariant) and monotonicity of $\arrowA$ imply that if after $\pi$ we encounter a formula beginning with $[\letterA]$, $\eve$ just uses $\pi\rho\mapsto\sigma_i(\pi'\rho)$ and win. If after $\pi$ we visit a formula beginning with $\diamond{\letterA}$ and $\sigma_i(\pi')=\point_{j''}$ for some $j''<\omega$, then either:
\begin{enumerate}
  \item $j'' < j$, hence $\point_j \arrowA \point_{j''}$ and $\eve$ wins using $\pi\rho\mapsto\sigma_i (\pi'\rho)$, or
  \item $j\leq j''$, which combined with item \ref{Invariant-Item-2} of the invariant gives $k+N \leq j \leq j''$. Since we have just made a modal move, we are left with at most $k-1$ possible modal moves in $\phaseScalar$, so the choice of $\point_{j-1}$ preserves the invariant, as $k-1+N \leq j-1\leq j''$ and $j-1<j<i$.
\end{enumerate}

Since $\phi$ is guarded, the maximal number of consecutive $\epsilon$-transitions in a play is bounded by $|\SubFor(\phi)|$. Moreover, each time a play passes through $\diamond{\letterA}$, the number $k$ decreases. As a result, after at most $k_\max\cdot|\SubFor(\phi)|$ moves we either end the game in $\phaseScalar$ or leave $\phaseScalar$. In the first case, thanks to item \ref{Invariant-Item-1} of the invariant, $\eve$ must win, for the strategy $\sigma$ is winning. The second case can happen by either (i) visiting a formula that begins with $\boxmodal{\letterA}, \boxmodal{\letterB}$ or $\diamond{\letterB}$ (in which case $\eve$ wins, as described above), or (ii) visiting a subformula $\psi$ s.t. $\rank(\psi)$ was visited in the pre-modal phase. But since $\phi$ is scalar (and so by Proposition \ref{Prop-Scalar=Injectively-Ranked} injectively ranked when seen as an automaton), this implies that \emph{the same} $\psi$ must have been visited in the pre-modal phase.
Denote by $\pi_1$ the play ending with the first visit to $\psi$ after the pre-modal phase (before the counter update, i.e. $\pi_1$ ends with a countdown configuration) and by $\pi_1'$ the corresponding $\sigma_i$-play that exists thanks to the invariant. Let $\pi_0$ and $\pi'_0$ be the prefixes of $\pi_1$ and $\pi'_1$, respectively, ending with the first visit in $\psi$ \emph{after the counter update} (i.e. $\pi_0$ and $\pi_0'$ end with a positional configuration). By item \ref{Invariant-Item-1} of the invariant, the plays $\pi_0$ and $\pi'_0$ ($\pi_1$ and $\pi'_1$) lead to the same counter assignment $\ctr_0$ ($\ctr_1$, respectively).

Consider the strategy $\sigma$ behaving as \emph{$\sigma_i$ after $\pi'_0$}, that is $\sigma(\rho) = \sigma_i(\pi'_0\rho)$ for every $\rho$. Then:
\begin{align}\label{Claim-Winning}
  \text{The strategy $\sigma$ is winning from $\pconf{(\point_j,\psi), \ctr_0}$ for every $N\leq j\leq i$.}
\end{align}
Indeed, for $j=i$, $\sigma$ just continues a $\sigma_i$-play, and hence leads to victory. 
For $N\leq j<i$, we also essentially just use the same strategy $\sigma$:
\begin{enumerate}
  \item The pre-modal phase starting from $(\point_j,\psi)$ is identical as if we started from $(\point_i,\psi)$ (recall that we identify pre-modal plays starting in different $(\point_i,\psi)$ and $(\point_j,\psi)$ if they are equal up to swapping positions $(\point_i,\theta)$ and $(\point_j,\theta)$ for all $\theta$).
  \item If after a pre-modal play $\rho$ the game ever reaches a formula $\theta$ that begins with a modal operator, $\eve$ may legally continue using $\sigma$. Indeed, since $\rho$ is pre-modal, it does not change the point, meaning that it leads from $(\point_j,\psi)$ to $(\point_j,\theta)$ and from $(\point_i,\psi)$ to $(\point_i,\theta)$. Because $\point_i$ and $\point_j$ have the same $\arrowB$-successors, if $\theta$ begins with $\diamond{\letterB}$ or $\boxmodal{\letterB}$, then the possible moves from the position $(\point_j,\theta)$ are the same as from $(\point_i,\theta)$ and so we can continue as if we started from $(\point_i,\theta)$.

  Similarly, since $j\leq i$ implies that every $\arrowA$-successor of $\point_j$ is an $\arrowA$-successor of $\point_i$, $\sigma$ can be used to win against every $\adam$'s choice of an $\arrowA$-successor of $\point_j$ if $\theta$ begins with $\boxmodal{\letterA}$.

  The remaining case is when $\theta$ begins with $\diamond{\letterA}$. By Proposition \ref{Prop-Bounded-A-Diamond}, if in the first modal step of a $\sigma_{\point_i}$-play $\eve$ has to provide an $\arrowA$-successor $\point_k$ of $\point_i$, then $\sigma_{\point_i}$ chooses some $\point_k$ with $k<N$. Since $\sigma_i = \sigma_{\point_i} +_\exists 1$, the same is true about $\sigma_i$. But since $\sigma(\rho)=\sigma_i(\pi'_0\rho)$ and both $\rho$ and $\pi'_0$ are pre-modal, it follows from $N\leq j$ that the choice given by $\sigma_i$ is legal from $\point_j$.
\end{enumerate}
This proves \eqref{Claim-Winning}.

\medskip


Note that since by definition in $\pi_1$ (and $\pi'_1$) it is the first time the game revisits a rank seen in the pre-modal phase, we have:
\begin{align}\label{Eq-Ctr-0-1}
  \ctr_0(r')=\ctr_1(r')
\end{align}
for every nonstandard rank $r'\geq r$.
Indeed, since we are in a game for a \emph{scalar} formula $\phi$, every superformula $\theta$ of $\psi$ must have been visited in the pre-modal phase. Since $\pi_1$ is a minimal play in which some rank is visited twice, guardedness of $\phi$ implies that no such $\theta$ was visited between $\pi_0$ and $\pi_1$. This means that the only formulas that appeared between $\pi_0$ and $\pi_1$ were \emph{strict subformulas} of $\psi$, and hence \emph{all the ranks} visited between $\pi_0$ and $\pi_1$ were strictly lower than $r$, which implies \eqref{Eq-Ctr-0-1}.

\medskip

To finish the proof of \eqref{Impl-Scalar} we need to show how to win once $\pi_1$ has been played and the game reached a countdown configuration $\cconf{(\point_j,\psi),\ctr_1}$. By item \ref{Invariant-Item-1} of the invariant, $N<j\leq i$. Consider the following cases:
\begin{itemize}
  \item If $r$ is standard, the counter update in $\cconf{(\point_j,\psi),\ctr_1}$ is deterministic and by \eqref{Eq-Ctr-0-1} leads to $\ctr_0$. Hence, we end up in a configuration $\pconf{(\point_j,\psi), \ctr_0}$. By \eqref{Claim-Winning}, $\eve$ may use $\sigma$ to win from there.
  
  \item If $r$ is nonstandard and belongs to $\adam$, then by \eqref{Eq-Ctr-0-1} it follows that for every $\ctr$ that $\adam$ can choose in $\cconf{(\point_j,\psi),\ctr_1}$, $\ctr\worse_\forall\ctr_0$. Thus, again by \eqref{Claim-Winning}, $\eve$ may win using $\sigma$ from every such $\pconf{(\point_j,\psi), \ctr}$.
  
  \item If $r$ is nonstandard and belongs to $\eve$, the choice $\ctr_0$ that was picked by $\sigma_i$ at the end of $\pi_0$ is not legal after $\pi_1$.
  However, since $\sigma_i$ is one above $\sigma_{\point_i}$, there exist $\sigma_{\point_i}$-plays $\pi_0^-$ and $\pi_1^-$ one below $\pi'_0$ and $\pi'_1$, respectively, ending with configurations $\ctr_0^-$ and $\ctr_1^-$ such that $\ctr_0^-+_\exists 1=\ctr_0$ and $\ctr_1^-+_\exists 1=\ctr_1$. Define a counter assignment $\ctr$:
  \[
    \ctr(r')=
    \begin{cases}
      \ctr_1(r') & \text{ if $r'>r$,}\\
      \ctr_1^-(r) & \text{ if $r'=r$,}\\
      \bounds(r') & \text{ if $r'<r$}.
    \end{cases}
  \]
  Such $\ctr$ is a legal update from $\ctr_1$, because $\ctr_1^-+_\exists 1=\ctr_1$ implies $\ctr_1^-(r)+1=\ctr_1(r)$. Such a choice leads to the configuration $\pconf{(\point_j,\psi), \ctr}$, and so we show that this configuration is winning for $\eve$. We have:
  \begin{align}\label{Ineq-Chosen-Better}
    \ctr^-_0\worse_\exists\ctr.
  \end{align}
  
  Indeed, \eqref{Eq-Ctr-0-1} implies that $\ctr^-_0(r')=\ctr^-_1(r')$ for all $r'\geq r$. Thus, for $r'>r$:
  \begin{align*}
    \ctr_0^-(r')=\ctr_1^-(r') \leq \ctr_1(r')=\ctr(r') & \text{\ \ \ \ \ \ if $r'$ belongs to $\eve$,}\\
    \ctr_0^-(r')=\ctr_1^-(r') = \ctr_1(r')=\ctr(r') & \text{\ \ \ \ \ \ otherwise;}
  \end{align*}

  and:
  \[
    \ctr_0^-(r)=\ctr_1^-(r)=\ctr(r);
  \]
  whereas for $r'<r$:
  \[
    \ctr^-_0(r')=\bounds(r')=\ctr(r').
  \]

  Note that since $\sigma_i=\sigma_{\point_i}+_\exists 1$ and $\sigma$ is defined as $\rho\stackrel{\sigma}{\longmapsto}\sigma_i(\pi'_0\rho)$, it follows that $\sigma=\sigma^-+_\exists 1$ with $\sigma^-$ given by $\rho\stackrel{\ \sigma^-}{\longmapsto}\sigma_{\point_i}(\pi'_0\rho)$. By \eqref{Claim-Winning}, $\sigma$ wins from $\pconf{(\point_j,\psi), \ctr_0}$. Consequently, $\ctr_0=\ctr_0^-+_\exists 1$ implies that $\sigma^-$ wins from $\pconf{(\point_j,\psi), \ctr_0^-}$, and thanks to \eqref{Ineq-Chosen-Better}, also from $\pconf{(\point_j,\psi), \ctr}$.
\end{itemize}
This completes the proof of \eqref{Impl-Scalar} from Proposition \ref{Prop-Bound-ArrowA}.\\

It remains to prove Proposition \ref{Prop-Bound-ArrowA}, i.e.~to refine strategies $\sigma_{\point_i}$ to obtain a finite bound $k_\max<\omega$ on the number of modal moves in a play in the phase $\phaseScalar$.
We will show a stronger fact: no play $\pi\in\phaseScalar$ visits the same formula of shape $\diamond{\letterA}\psi$ twice. Then, Proposition~\ref{Prop-Bound-ArrowA} follows with the bound $k_\max=|\SubFor(\phi)|+1$, because all the other modal moves (i.e. moves corresponding to formulas beginning with $\boxmodal{\letterA}$, $\diamond{\letterB}$ or $\boxmodal{\letterB}$) end $\phaseScalar$ immediately.

Before we go into the somewhat technical details, let us sketch the core idea of the proof, which splits into two steps. First, we show that if instead of updating the counters during $\phaseScalar$ the players only decrement them once upon leaving $\phaseScalar$, this does not change the winner of the game.
Second, we use this equivalence to massage $\sigma_{\point_i}$ so that instead of performing a sequence:
\[
  (\point_j,\diamond{\letterA}\psi)\to(\point_{j'},\psi)\to\ ...\ \to(\point_l,\diamond{\letterA}\psi)\to(\point_{l'},\psi) \in V^+
\]
of modal moves corresponding to $\diamond{\letterA}\psi$, $\eve$ immediately goes to the \emph{last} point $(\point_j,\diamond{\letterA}\psi)\to(\point_{l'},\psi)$. This is possible thanks to transitivity and well-foundedness of $\arrowA$ and avoids repetitions of $\diamond{\letterA}\psi$.

To prove the claim, it is enough if for every \emph{minimal} (and therefore necessarily ending with a first modal move) $\pi_I\in\phaseScalar$ we refine $\sigma_{\point_i}$ to a strategy $\sigma_{\pi_I}$ so that:
\begin{enumerate}
  \item $\sigma_{\pi_I}$ does not visit any $\diamond{\letterA}\psi$ twice in any play $\rho\in\phaseScalar_{\pi_I}$ and
  \item the behaviour on all other plays is not changed, meaning that
  $
    \sigma_{\pi_I}(\rho) = \sigma_{\point_i}(\rho)
  $
  for every $\rho$ without a prefix in $\phaseScalar_{\pi_I}$.
\end{enumerate}
If we do that for every minimal $\pi_I\in\phaseScalar$, we may combine all such refined strategies into one:
\[
  \sigma_\phaseScalar(\rho) =
  \begin{cases}
    \sigma_{\pi_I}(\rho) & \text{if $\pi_I$ is the minimal prefix of $\rho$ from $\phaseScalar$},\\
    \sigma_{\point_i}(\rho) & \text{otherwise (i.e. for pre-modal $\rho$)};
  \end{cases}
\]
that avoids repetitions of each $\diamond{\letterA}\psi$ in \emph{every} $\pi\in\phaseScalar$.

Towards such a refinement of $\sigma_{\point_i}$, fix a minimal $\pi_I\in\phaseScalar$ leading to a winning countdown configuration $\gamma=\cconf{(\point_z,\theta_z),\ctr_z}$. Denote $v_z=(\point_z,\theta_z)\in V$ and $\phaseScalarBis=\{\rho\ |\ \pi_I\rho\in\phaseScalar\}$. To get our desired $\sigma_{\pi_I}$ it suffices to construct a winning strategy for $\game,\gamma$ that avoids repetitions of each $\diamond{\letterA}\psi$ in every $\pi\in\phaseScalarBis$.

Note that mebership in $\phaseScalarBis$ only depends on the underlying positions. Let $V_\phaseScalarBis\subset V$ be the set of all the positions of shape $(\point,\xi)$ with $\xi$ either (i) beginning with $\boxmodal{\letterA}$, $\boxmodal{\letterB}$ or $\diamond{\letterB}$ or (ii) having a rank that was visited in $\pi_I$. Then $\pi\in\phaseScalarBis$ iff in $\pi$ no position other than the last one belongs to $V_\phaseScalarBis$.

Define a \emph{parity} game  $\widetilde{\game}$ that has arena $(V,E,\rank)$ with all the positions from $V_\phaseScalarBis$ turned into terminal positions immediately winning for $\eve$. To avoid confusion, we will call parity plays $\overline{v}\in V^*$ in $\widetilde{\game}$ \emph{paths} and reserve the term \emph{plays} for $\game$. Observe that for every $\overline{v}\in V^*$:
\begin{align}\label{equiv:pos[G]~parity(G)}
  \text{$\overline{v}$ is a path in $\widetilde{\game},v_z$ $\iff$ $\overline{v}=\positions(\pi)$ for some play $\pi\in\phaseScalarBis$ in $\game,\gamma$}
\end{align}
with the left to right implication following from the fact that all the counters decremented in $\phaseScalarBis$ have initial, and hence limit values in $\gamma$. In particular, \eqref{equiv:pos[G]~parity(G)} implies that every $\phaseScalarBis$-component of a winning strategy for $\game,\gamma$ is a winning strategy for $\widetilde{\game},v_z$. This justifies the following terminology: we call a partial function $f:V^*\to V$, thought of as a candidate for a $\phaseScalarBis$-component of a winning strategy for $\game,\gamma$, a \emph{proto-strategy} if $f$ is a winning strategy in $\widetilde{\game},v_z$.

Since every modal move over $\arrowB$ leaves $\phaseScalarBis$, it follows that all the positions accessible in $\pi\in\phaseScalarBis$ are of the form $(\point_i,\psi)$ for $i\leq z$ (because $\arrowA$ only leads to points with a strictly smaller index) and no such position repeats in $\phaseScalarBis$ (by guardedness of $\phi$ and acyclicity of $\arrowA$). It follows that the set $\positions[\phaseScalarBis]$ is finite. By \eqref{equiv:pos[G]~parity(G)}, this means that also paths in $\widetilde{\game},v_z$ are all finite. Hence, if $f$ is a proto-strategy then every maximal $f$-path in $\widetilde{\game},v_z$ ends with a position $v$ that either (i) belongs to $V_\phaseScalarBis$ or (ii) is controlled by $\adam$ and has no successors in both $\game$ and $\widetilde{\game}$.

We prove that for every proto-strategy $f$, the following are equivalent:
  \begin{enumerate}
    \item\label{it:games 1} $\eve$ has a winning strategy $\sigma$ for $\semanticGame,\gamma$ 
    guided by $f$ in $\phaseScalarBis$.
    \item\label{it:games 2} $\eve$ wins in the following game:
    
    (i) $\adam$ picks a maximal $f$-path $\overline{v}\in V^*$ starting at $v_z$;
    
    (ii) we play a usual countdown game from $\gamma$ 
    but on arena restricted to $\overline{v}$ (i.e. we only update the counters and deterministically follow $\overline{v}$);
    
    (iii) $\eve$ wins iff the resulting configuration $\cconf{v, \ctr}$, with $v$ being the last position of $\overline{v}$, is winning for $\eve$ in $\game$.
    \item\label{it:games 3} $\eve$ wins in the following game:
    
    (i) $\adam$ picks a maximal $f$-path $\overline{v}\in V^*$ starting at $v_z$;
    
    (ii) given the set $\nonstandardRanks_{\overline{v}}\subset\nonstandardRanks$ of all the nonstandard ranks that should have non-initial value after traversing $\overline{v}$, the owner of each $r\in\nonstandardRanks_\pi$ (starting from more important ranks) picks a final counter value $\ctr(r)<\ctr_z(r)$ and we put $\ctr(r)=\ctr_I(r)$ for all other $r \in \nonstandardRanks-\nonstandardRanks_{\overline{v}}$;
    
    (iii) $\eve$ wins iff the resuting configuration $\cconf{v, \ctr}$, with $v$ being the last position of $\overline{v}$, is winning for $\eve$ in $\game$.
  \end{enumerate}

  Note that the set $\nonstandardRanks_{\overline{v}}$ in \eqref{it:games 3} is uniquely determined by $\overline{v}$, as $r\in\nonstandardRanks_{\overline{v}}$ iff $r$ appears in $\overline{v}$ and no higher rank appears after the last occurrence of $r$. However, since we are dealing with a game corresponding to a scalar formula $\phi$, $\nonstandardRanks_{\overline{v}}$ has an even more straightforward description: nonstandard $r$ belongs to $\nonstandardRanks_{\overline{v}}$ iff it is a rank of some superformula of the last formula in $\overline{v}$.\\

  The implication \eqref{it:games 1}$\implies$\eqref{it:games 2} is straightforward. Once $\adam$ picked $\overline{v}$, $\eve$ simply uses $\sigma$ until $\overline{v}$ is traversed. By \eqref{equiv:pos[G]~parity(G)}, until that moment the game stays in $\phaseScalarBis$, so the choices dictated by $\sigma$ are consistent with $\overline{v}$, as $\sigma$ is guided by $f$ in $\phaseScalarBis$ and $\overline{v}$ is an $f$-path. Since $\sigma$ is winning, the configuration reached at the end of $\overline{v}$ must be winning.

  To prove that \eqref{it:games 2}$\implies$\eqref{it:games 1}, assume that for every maximal $f$-path $\overline{v}$ starting at $v_z$, $\eve$ has a strategy $h_{\overline{v}}$ winning in the second stage of \eqref{it:games 2}. Our goal is to provide her with $\sigma$ for \eqref{it:games 2}. When during $\phaseScalarBis$ she has to pick an edge, she uses $\sigma(\rho)=f(\positions(\rho))$ so that $\sigma$ is guided by $f$.
  For choosing ordinals, observe that the tree of all paths in $\widetilde{\game},v_z$ is finite, so for every countdown play $\rho$ guided by $f$ there are only finitely many maximal $f$-paths extending $\positions(\rho)$. Thus, for every play $\rho$ ending in $\eve$'s choice of a counter for rank $r$, she takes the ordinal:
  \[
    \max \{h_{\overline{v}}(\rho)(r)\ |\ \overline{v} \text{ is a maximal $f$-path extending $\positions(\rho)$} \}
  \]
  which is legal, since the longer $\rho$ is, the fewer paths extend $\positions(\rho)$. This way, she either wins before $\phaseScalarBis$ ends, or leave it in a winning configuration, and in the later case she may continue with any winning strategy.

  It remains to prove that  \eqref{it:games 2}$\iff$\eqref{it:games 3}. Note that in \eqref{it:games 2}, once the path $\overline{v}$ is chosen, the only nontrivial choice of a value for $r$ is when its the corresponding counter has initial value and is not going to be reset further in $\overline{v}$.

  Indeed, without lost of generality successor vaues are always decremented by $1$ and if $r$ is going to be reset somewhere further in $\overline{v}$, it suffices to pick the number $k$ of visits in $r$ before the closest reset.
  After the last reset of the counter for $r$, the number $k$ of future decrements of $r$ in $\overline{v}$ is fixed, so in order to end the game with $\ctr(r)=\alpha$ it suffices to pick the value $\alpha+k$ (and again, decrement it by 1 each time the game visits $r$).
  The above choices are legal because by definition of $\phaseScalar$, $r$ was not visited in the pre-modal phase and hence its counter has a maximal value upon entering $\phaseScalar$.

  Moreover, the order of these \emph{nontrivial} choices is precisely the (decreasing) order on $\nonstandardRanks_{\overline{v}}$. This establishes an equivalence between \eqref{it:games 2} and \eqref{it:games 3}, therefore completing the proof of equivalence of games \eqref{it:games 1}, \eqref{it:games 2} and \eqref{it:games 3}.\\

  Let $\sigma$ be a winning strategy for $\game,\gamma$. To complete the proof of Proposition~\ref{Prop-Bound-ArrowA} it suffices to upgrade such $\sigma$ so that no formula component of shape $\diamond{\letterA}\theta$ repeats in $\phaseScalarBis$. Thanks to finiteness of $\positions[\phaseScalarBis]$ we may apply Proposition \ref{Prop-Finite-CtrIndep} and assume that $\sigma$ is guided by $\sigma^\phaseScalarBis$ in $\phaseScalar$. As mentioned, such $\sigma^\phaseScalarBis$ is a legal poto-strategy. Enumerate all the subformulas $\diamond{\letterA}\psi_1,...,\diamond{\letterA}\psi_n$ of $\phi$ of shape $\diamond{\letterA}\theta$. We construct, by induction on $i\leq n$, a sequence $f_0, ..., f_n:V^*\to~V$ of proto-strategies s.t.:
  \begin{enumerate}
    \item $f_0=\sigma^\phaseScalarBis$,
    \item whenever $i<j$ and $\overline{v}\in V^*$ is a maximal $f_j$-path, there exists a maximal $f_i$-path $\overline{w}\in V^*$ ending with the same position,
    \item $f_i$ avoids repetitions of $\{\psi_1,...,\psi_i\}$.
  \end{enumerate}
  
  Assume we already have $f_i$ and want to construct $f_{i+1}$. For every $f_i$-path $\overline{v}\in V^*$ ending in a visit in $\diamond{\letterA}\psi_{i+1}$ consider the set:
  \[
    \HH_{\overline{v}} = \{\overline{w}\in V^*\ |\ \text{$\overline{w}$ is a $f_i$-path, has $\overline{v}$ as a prefix and ends with $\diamond{\letterA}\psi_{i+1}$} \}
  \]
  and fix some $(\overline{v})^\circ$ maximal in $\HH_{\overline{v}}$ ($\HH_{\overline{v}}$ is nonempty as it contains $\overline{v}$ and must contain a maximal path because the length of paths is bounded).
  
  Our new strategy $f_{i+1}$ acts like $f_i$ until the first visit in $\diamond{\letterA}\psi_{i+1}$ and then, instead of making multiple $\arrowA$-moves for $\diamond{\letterA}\psi_{i+1}$, immediately jumps to the \emph{last} choice from such maximal $(\overline{v})^\circ$ that extends the current path $\overline{v}$:
  \[
    f_{i+1}(\overline{w})=
    \begin{cases}
      f_i(\overline{w}) & \text{ if $\overline{w}$ does not visit $\diamond{\letterA}\psi_{i+1}$,}\\
      f_i((\overline{v})^\circ\cdot\overline{u}) & \text{ otherwise, with $\overline{w}=\overline{v}\cdot\overline{u}$ and $\overline{v}$ ending in the first visit in $\diamond{\letterA}\psi_{i+1}$}.
    \end{cases} 
  \]
  Such $f_{i+1}$ is a legal proto-strategy. Indeed, the new moves are allowed thanks to transitivity of $\arrowA$ and $f_{i+1}$ is winning in $\widetilde{\game},v_z$, because positions accessible via $f_{i+1}$ are a subset of the ones accessible by $f_i$. This also implies the second property, whereas the third one follows from the fact that each $(\overline{v})^\circ$ is maximal in $\HH_{\overline{v}}$.

  Since by scalarity of $\phi$ the set $\nonstandardRanks_{\overline{v}}$ in the third variant of the game \eqref{it:games 3} depends only on the last formula in $\overline{v}$ and $\eve$ wins \eqref{it:games 3} with $f=f_0$, thanks to the second property she also wins \eqref{it:games 3} with $f_n$. By equivalence of \eqref{it:games 3} and \eqref{it:games 1}, this means that some strategy $\sigma_{\pi_I}$ winning from $\gamma$ is guided by $f_n$. Moreover, the third property implies that $\sigma_{\pi_I}$ avoids repetitions of all $\diamond{\letterA}\psi_1, ..., \diamond{\letterA}\psi_n$ in $\phaseScalarBis$, thus proving Proposition \ref{Prop-Bound-ArrowA} and completing the proof of Theorem \ref{Thm-Vectorial-Vs-Scalar}.

%% file: appendix-strictness.tex
 \section{Proof of Lemma~\ref{Lem-Strictness}}\label{app:strictness}

\begin{proposition}\label{Prop-Finite-Alternation}
  For every countdown formula $\phi$ there is a finite constant $t_\phi<\omega$ such that for every valuation $\val$ stable above $\kappa$, in the part $[\kappa,\omega_1)$ of the model above $\kappa$, $\phi$ changes its truth value at most $t_\phi$ times.
\end{proposition}

\begin{proof}
Since without loss of generality the formula is guarded (see Proposition~\ref{Prop-Guardedness}), by Proposition \ref{Prop-PreModal-CtrIndep} we may assume that in the semantic game $\eve$ always uses a pre-modally counter-independent strategy. But the number $z$ of possible pre-modal components for such strategies is finite, so if $\phi$ changed its value more than $t_{\phi}=2z+2$ times above $\kappa$, there would be $\kappa\leq\alpha<\zeta<\beta$ such that $\eve$ wins from $\alpha$ and $\beta$ with the same pre-modal component, but loses from $\zeta$ in between, which is impossible by Proposition \ref{Prop-Convex-Meaning-On-Monotone}.
\end{proof}

We prove Lemma~\ref{Lem-Strictness} by induction on the complexity of the formula $\phi$. The base case is immediate, as for every $x \in \Var$ it suffices to take $\alpha_x=0$. For propositional connectives and modal operators we take $\alpha_{\psi_1 \lor \psi_2}=\alpha_{\psi_1\land\psi_2}=\max(\alpha_{\psi_1},\alpha_{\psi_2})$ and $\alpha_{\Diamond\psi}=\alpha_{\Box\psi}=\alpha_\psi+1$. The remaining non-trivial cases are countdown and fixpoint operators.
\begin{itemize}
  \item For $\phi=\eta^\omega_i\overline{x}.\overline{\psi}$,
  let $\Phi=\{\theta_1,...,\theta_l\}$ be the set of all maximal subformulas of $\overline{\psi}$ not using any variable $x_j$. For each $\theta$, pick a fresh variable $y_{\theta}$ and put:
  \[
    \psi'_j=\psi_j[\theta_1\mapsto y_{\theta_1}...\theta_l\mapsto y_{\theta_l}]
  \]
  i.e. starting from the root $\psi_j$, we replace every subformula $\theta$ that has no variables from $\overline{x}$ with a fresh variable $y_{\theta}$.\footnote{Recall that we do not identify isomorphic subformulas, and so there are no substitutions \emph{inside} the $\theta$'s. In particular, the order of substitutions does not matter.} Observe that $\psi_j=\psi'_j[y_{\theta_1}\mapsto \theta_1...y_{\theta_l}\mapsto \theta_l]$, so: 
    \[
    \eta^\omega_i \overline{x}.\overline{\psi} \equiv (\eta^\omega_i \overline{x}.\overline{\psi'})[y_{\theta_1}\mapsto \theta_1...y_{\theta_l}\mapsto \theta_l].
  \]
  Note that if $\phi$ has countdown nesting at most $k$, then each $\psi'_j$ and each $\theta$ has countdown nesting less than $k$. Thus, by the induction hypothesis there exist $\alpha_{\psi'_j}<\omega^{k}$ and $\alpha_\theta<\omega^{k}$ s.t. $\psi'_j$ and $\theta$ stabilize $\alpha_{\psi'_j}$ and $\alpha_\theta$ above the valuation, respectively. Denote $\alpha_{\overline{\psi'}}=\max\{\alpha_{\psi'_1},...,\alpha_{\psi'_n}\}$.
  
  For $m<\omega$, consider the $m$-th unfolding given by $\psi_j'^0 =x_j$ and $\psi_j'^{m+1}=\psi_j'[x_1 \mapsto \psi_1'^m ... x_n \mapsto \psi_n'^m]$. It follows by a straightforward induction on $m$ that each $\psi_j'^m$ is stable $\alpha_{\overline{\psi'}}\times m$ above the valuation. Moreover, for any valuation $\val$ we have:
  $$\textstyle
    \semantics{\mu^\omega_j\overline{x}.\overline{\psi'}}^\val=\bigcup_{m<\omega}\semantics{\psi_j'^m}^\val \qquad \text{ and } \qquad \semantics{\nu^\omega_j\overline{x}.\overline{\psi'}}^\val=\bigcap_{m<\omega}\semantics{\psi_j'^m}^\val
  $$
  so $\eta^\omega_j\overline{x}.\overline{\psi'}$ is stable $\alpha_{\overline{\psi'}}\times \omega$ above the valuation. Finally, we obtain that $\phi=(\eta^\omega_j\overline{x}.\overline{\psi'})[\theta_1\mapsto y_{\theta_1}...\theta_l\mapsto y_{\theta_l}]$ is stable $\alpha_\phi$ above valuation with:
  \[
    \alpha_\phi = \max\{\alpha_{\theta_1}...\alpha_{\theta_l}\} + \alpha_{\overline{\psi'}} \times \omega.
  \]
  Since $\alpha_{\overline{\psi'}}\times \omega < \omega^{k+1}$ and for each $\theta$, $\alpha_\theta<\omega^{k}$, it follows that $\alpha_\phi<\omega^{k+1}$.

  \item For $\phi=\eta^\infty_i\overline{x}.\overline{\psi}$, note that the countdown nesting of each $\psi_j$ is not greater than that of $\phi$. For each $j\leq n=|\overline{x}|$, let $t_{\eta^\infty_j\overline{x}.\overline{\psi}}<\omega$ be the constant from Proposition \ref{Prop-Finite-Alternation} and $\alpha_{\psi_j}<\omega^{k+1}$ the constant that exists by the inductive hypothesis. Put $\alpha_\max=\max_{j\leq n} (\alpha_{\psi_j})$, $t_\max=\max_{j \leq n}(t_{\eta^\infty_j\overline{x}.\overline{\psi}})$ and $\alpha_\phi=\alpha_\max \times t_\max \times n$.
  Clearly $\alpha_\phi<\omega^{k+1}$, as $\alpha_\max<\omega^{k+1}$ and $t_\max<\omega$ -- so it suffices to show that such bound works. Define a valuation:
  \[
    \val'(y)=
    \begin{cases}
      \semantics{\eta^\infty_j\overline{x}.\overline{\psi}}^\val & \text{ if $y=x_j$}\\
      \val(y) & \text{otherwise.}
    \end{cases}
  \]
  and let $\kappa$ be the stabilization point of $\val$. Note that for each $j\leq n$, $\val'(x_j)$ changes value above $\kappa$ at most $t_\max$ times and so $\val'$ changes its value at most $t_\max \times n$ times above $\kappa$.
  
  On the other hand, if $\val'$ does not change its value for at least $\alpha_\max$ steps, it remains constant forever, i.e. if for some $\kappa\leq\alpha<\omega_1$ we have that $\val'$ is constant on the interval $[\alpha,\alpha+\alpha_\max]$, then it is constant on the entire $[\alpha,\omega_1)$. Indeed, assuming that $\val'$ is constant on $[\alpha,\alpha+\alpha_\max]$, we show by induction on $\alpha_\max\leq\beta$ that it is constant on $[\alpha,\alpha+\beta]$. Indeed we have:
  \[
    \val'(x_j) = \semantics{\eta^\infty_j\overline{x}.\overline{\psi}}^\val = \semantics{\eta^\infty_j\overline{x}.\overline{\psi}}^{\val'} = \semantics{\psi_j}^{\val'}
  \]
  and since all $\overline{x}$ are guarded in $\overline{\psi}$, $\semantics{\psi_j}^{\val'}(\alpha+\beta)$ depends only on the values of $\val'$ strictly below $\alpha+\beta$. In particular, for every $\zeta\leq \alpha+\beta$ we have $\semantics{\psi_j}^{\val'}(\zeta)=\semantics{\psi_j}^{\val_{\alpha+\beta}}(\zeta)$ where $\val_{\alpha+\beta}$ is the valuation that repeats the last value above $\alpha+\beta$:
  \[
    \val_{\alpha+\beta}(y)(\zeta)=
    \begin{cases}
      \val'(y)(\zeta) & \text{if $\zeta<\alpha+\beta$}\\
      \val'(y)(\alpha+\alpha_\max) & \text{otherwise.}
    \end{cases}
  \]
  By the inductive hypothesis, $\val'$ is constant on $[\alpha, \alpha+\beta)$, so $\val_{\alpha+\beta}$ is constant on $[\alpha,\omega_1)$, i.e. it stabilizes at $\alpha$. This implies that $\semantics{\psi_j}^{\val_{\alpha+\beta}}$ is stable above $\alpha+\alpha_{\psi_j}$. Since $\alpha_{\psi_j}\leq\alpha_\max\leq\beta$, we get that:
  \[
    \semantics{\psi_j}^{\val'}(\alpha+\beta) =  \semantics{\psi_j}^{\val_{\alpha+\beta}}(\alpha+\beta)
     =  \semantics{\psi_j}^{\val_{\alpha+\beta}}(\alpha+\alpha_\max)
     =  \semantics{\psi_j}^{\val'}(\alpha+\alpha_\max)
  \]
  
  which shows that $\val'$ is indeed constant on $[\alpha,\alpha+\beta]$.

  It follows that after at most $t_\max \times n$ blocks, each of length at most $\alpha_\max$, the valuation $\val'$ stabilizes.
\end{itemize}
This finishes the proof of Lemma \ref{Lem-Strictness} and Theorem \ref{Thm-Strictness-Nesting}.

%
%